\documentclass[11pt]{article}
\usepackage{amssymb}

\usepackage{graphicx}
\usepackage{amsmath}
\usepackage{makeidx}
\usepackage{indentfirst}

\newcounter{resultnum}[section]

\newcounter{conclusionnum}[section]

\newcounter{conditionnum}[section]

\newcounter{conjecturenum}[section]

\newcounter{examplenum}[section]

\newcounter{exercisenum}[section]
\newtheorem{lemma}{Lemma}[section]

\newcounter{lemmanum}[section]

\newcounter{notationnum}[section]
\newtheorem{theorem}{Theorem}[section]

\newcounter{theoremnum}[section]

\newcounter{definitionnum}[section]
\newtheorem{corollary}{Corollary}[section]

\newcounter{corollarynum}[section]

\newcounter{remarknum}[section]

\newcounter{propositionnum}[section]

\newcounter{acknowledgementnum}[section]

\newcounter{algorithmnum}[section]

\newcounter{axiomnum}[section]

\newcounter{casenum}[section]

\newcounter{claimnum}[section]

\newcounter{summarynum}[section]

\newcounter{problemnum}[section]
\newenvironment{proof}[1][]{\textbf{Proof.} }{}

\begin{document}

\title{Metric Compatible  or Noncompatible\\ Finsler--Ricci Flows}
\date{April 14, 2012}
\author{\textbf{Sergiu I. Vacaru} \thanks{%
sergiu.vacaru@uaic.ro,\ http://www.scribd.com/people/view/1455460-sergiu
\newline
All Rights Reserved \copyright \ 2011 \ Sergiu I. Vacaru} \and \textsl%
{\small Science Department, University "Al. I. Cuza" Ia\c si, } \\
\textsl{\small 54, Lascar Catargi street, Ia\c si, Romania, 700107 } }
\maketitle

\begin{abstract}
There were elaborated different models of Finsler geometry using the Cartan
(metric compatible), or Berwald and Chern (metric non--compatible)
connections, the Ricci flag curvature etc. In a series of works, we studied
 (non)commutative metric compatible Finsler and nonholonomic
generalizations of the Ricci flow theory
[see S. Vacaru,\ J. Math. Phys. \textbf{49} (2008) 043504;\
\textbf{50} (2009) 073503 and references therein]. The goal of this
work is to prove that there are some models of Finsler gravity and
geometric evolution theories with generalized Perelman's functionals,
and correspondingly derived nonholonomic Hamilton evolution
equations, when metric noncompatible Finsler connections are
involved. Following such an approach, we have to consider distortion
tensors, uniquely defined by the Finsler metric, from the Cartan and/or
the canonical metric compatible connections. We conclude that,
in general, it is not possible to elaborate self--consistent models
of geometric evolution with arbitrary Finsler metric noncompatible connections.
\end{abstract}


\section{Motivation and Introduction}

Geometric analysis and evolution equations are important topics of research
in modern mathematics and physics, see original R. Hamilton's \cite%
{ham1,ham2} and G. Perelman's \cite{gper1,gper2,gper3} works and reviews of
results in \cite{caozhu,kleiner,rbook}. In 2007, it was published a
communication at a Conference in memory of M. Matsumoto (at Sapporo, in
2005), where D. Bao \cite{bao1} mentioned that the idea to study such
problems related to Finsler geometry came to S. -S. Chern in 2004.
Unfortunately, the famous mathematician had not published  his
proposals/results on a Finsler--Ricci flow theory.\footnote{%
It was one--two years after famous Grisha Perelman's electronic preprints
containing the proof of the Thurston/ Poincar\`{e} conjecture were put in
arXiv.org. That induced a number of papers  on geometric flows
and applications related to various branches of mathematics, physics,
optimization etc. I'm grateful to D. Bao and E. Peyghan for important
correspondence, historical remarks about S. Chern original ideas, and
discussions on Finsler--Ricci flows and almost K\"{a}hler models of Finsler
geometry and generalizations.}

In May-June, 2005, there were a series of lectures of N. Higson at Madrid,
Spain, where the R. Hamilton and G. Perelman fundamental contributions in
mathematics were discussed with respect to possible applications in modern
gravity, cosmology and astrophysics. The author of this paper attended one
of those lectures at CSIC, Madrid. At that time, he worked in some
directions of nonholonomic mechanics and Finsler geometry and geometric
methods of constructing exact solutions in Einstein gravity and
modifications. He knew that Chern's connection in Finsler geometry is metric
noncompatible which gives rise to a number of difficulties for applications
related to standard theories of physics (see discussions in \cite%
{vrev1,vcrit,vaxiom,vsgg}; we also mention here some most important
monographs on Finsler geometry \cite{cartan,matsumoto,ma,bejancu,bcs}). It
is obvious that a general extension of the Hamilton--Perelman theory for
metric noncompatible spaces, including Finsler models, is not possible. If $%
\mathbf{Dg}\neq 0$ for a  metric $\mathbf{g} $ and a
linear connection $\mathbf{D}$ (such geometric objects may be Finsler or other types), the evolution of geometric objects on a real
parameter $\chi$ can not be determined only by a Ricci tensor (see relevant
formulas on next page and rigorous definitions in sections \ref{s3} and \ref%
{s4}).

In our works, we preferred to use the Cartan connection and metric
compatible modifications and generalizations of Finsler geometry because the
geometric constructions and proofs of the main results are quite similar to
those for Riemannian spaces but for some special classes of Finsler
connections. A series of results were developed for the theory of
nonholonomic Ricci flows (with additional non--integrable constraints) for
certain classes of Einstein, Finsler, Lagrange and other nonholonomic,
noncommutative, nonsymmetric, fractional and stochastic spacetimes and
geometries \cite%
{vric1,vric2,vric3,vric4,vric5,vric6,vric7,vric8,vric9,vric10}.

The problem of Ricci flows and Finsler geometry was considered again in a
recent paper \cite{peyghan}\footnote{%
I thank E. Peyghan for sending two preliminary versions of their work before
the authors would publish the results in a preprint or journal version},
where Finsler--Ricci flow type evolution equations are studied following D.
Bao's heuristic proposals related to geometric flows and Finsler
geometry. In such a case, even the Berwald connection (which is also metric
noncompatible) is involved, the constructions may be associated to the
Cartan metric compatible connection. A new definition/type of the Ricci
tensor \cite{akbar,bao1,bcs} which is symmetric and seem to provide an
alternative approach to formulating Finsler like gravity and Ricci flow
theories is considered. Such results are original and important.
Nevertheless, the geometric evolution equations with right side Ricci flag
curvature postulated in the mentioned works (by D. Bao and A. Tayebi and
E. Peyghan) were not derived from certain generalized Perelman's functionals.
It was  not clear if such equations may describe an evolution gradient process
(we shall prove this in the present paper, as a particular case). We also
note that it was not stated if, and when, the models of Finsler--Ricci
flows with flag curvature may have certain limits to standard Laplacian
operators and Levi--Civita configurations  - this would be an important
argument that such theories may describe well--defined evolution processes.

In this work we extend our former results on Finsler--Ricci flows for metric
compatible connections in a more general context when metric noncompatible
Finsler connections (like the Berwald and Chern ones) are used for
nonholonomic deformations of Perelman's functionals. We shall analyze possible
relations to former results on nonholonomic Ricci flows and
Lagrange--Finsler evolution models via Cartan type (metric compatible)
connections which positively describe geometric evolution processes in a
self--consistent and similar manner to the Ricci flow theory on Riemannian
manifolds.

R. S. Hamilton`s evolution equations were postulated for real Riemannian
manifolds \cite{ham1,ham2} following heuristic arguments,
 $$\frac{\partial g_{ij}}{\partial \chi } =-2Ric_{ij}, \qquad g_{ij}\mid _{\chi
=0} =\ ^{\circ }g_{ij}(x^{k}).$$
In these equations, geometric flows of metrics $g_{ij}(\chi ,x^{k})$ are
considered for a real parameter $\chi $ on a manifold $M$ when local
coordinates $x^{k}$ are labeled by indices $i,j,...=1,2,..., n=dimM$. The
Ricci tensor $Ric_{ij}$ in defined by the Levi--Civita connection $\nabla $
of $g_{ij}$ (for our purposes, it is enough to work with non--normalized
flows).

For a Finsler fundamental/generating function\footnote{%
see definitions and details in next section}, $F(x^{k},y^{a}),$ we can
consider
\begin{equation}
\ ^{v}{\tilde{g}}_{ij}=\frac{1}{2}\frac{\partial ^{2}F^{2}}{\partial
y^{i}\partial y^{j}}  \label{hessian}
\end{equation}%
as a ''vertical'' (v) metric on typical fiber if $\det |\ ^{v}{\tilde{g}}%
_{ij}|\neq 0.$\footnote{%
On $TM,$ we can identify the horizontal, $h$, and $v$--indices, i.e. $%
i,j,... $ and $a,b,...)$. In our work, left ''up'' and ''low'' indices are
used as labels, for instance, ''F'' being associated to ''Finsler'' etc. We
cite the monographs \cite{cartan,matsumoto,ma,bejancu,bcs} on main Finsler
geometry methods and comprehensive bibliography and our papers \cite%
{vrev1,vcrit}, for critical remarks, principles and perspectives of
applications in modern physics, cosmology and geometric mechanics. We
suggest readers to consult such works for reviews of results and notation
conventions.} Following a formal analogy to Hamilton's works, but for $\ ^{v}%
{\tilde{g}}_{ij}$ on tangent bundle $TM,$ we can postulate certain evolution
equations of type
\begin{equation}
\frac{\partial \ ^{v}{\tilde{g}}_{ij}}{\partial \chi }\sim \ ^{F}Ric_{ij}, \label{riccifinslerheur}
\end{equation}
where $\ ^{F}Ric_{ij}$ is a variant of Ricci tensor constructed for a model
of Finsler geometry and flows/evolution of fundamental Finsler functions are
parametr\-ized by $F(\chi ,x^{k},y^{a}).$ A heuristic definition of $\
^{F}Ric_{ij}$ is related to an important question if a chosen Finsler type
Ricci tensor would limit, or not, a Laplacian operator $\ ^{F}\Delta $
derived in metric compatible form for a Finsler geometry model. The answer
is affirmative for Laplacians determined by the Levi--Civita and/or Cartan
connections but not for models of Finsler geometry when $\ ^{F}Ric_{ij}$ is
introduced in a "nonstandard" form, or using a general metric noncompatible
Finsler connection. In \cite{akbar,bao1,peyghan}, the problem if and how a
Laplacian $\ ^{F}\Delta $ may be associated naturally to the Ricci flag (and
Akbar--Zadeh's) curvature and geometric flows was not analyzed.

The goal of this paper is to prove that (nonholonomically constrained)
Finsler--Ricci flow evolution equations and corresponding $\ ^{F}Ric_{ij}$
can be derived for some classes of metric noncompatible Finsler connections
and/or Akbar--Zadeh's Ricci curvature. If such geometric objects are
determined in unique forms (up to frame/coordinate transforms) by respective
distortion tensors which, in their turns, are also completely defined by a
Finsler fundamental function, we can formulate well defined Finsler
evolution theories. In our approach, we use our former results and techniques
elaborated for the models of geometric Finsler evolution with the Cartan
connection and certain metric compatible generalizations \cite%
{vric3,vric4,vric6}. Such constructions are very similar to those for
Riemannian spaces but derived with respective Finsler connections and
adapted frames. This allows us to define certain generalized Perelman
functionals and associated entropy and thermodynamical type values and
derive Hamilton type evolutions equations.

The paper is organized as follows. In Section \ref{s2} we survey the most
important geometric constructions and the basic language on metric
compatible Finsler spaces. In Section \ref{s3} there are defined the
fundamental geometric objects for metric noncompatible Finsler spaces
(using distortions from compatible ones) and
provided the most important formulas for Einstein--Finsler spaces.
The material outlined in the first three sections is oriented to
non--experts on Finsler geometry but researchers on geometric
analysis and mathematical physics. Perelman's
functionals are defined for special classes of metric noncompatible
Finsler spaces in Section \ref{s4}. There are proven main theorems on
Finsler--Ricci flows and nonholonomic, in general, metric
noncompatible geometric evolution equations. We also speculate
on statistical analogy and thermodynamics for Finsler--Ricci flows.

{\bf Acknowledgement:}\  I'm grateful to D. Bao and E. Peyghan for interest, discussions  and
correspondence on Finsler--Ricci flows. The research in this paper is partially supported by
the Program IDEI, PN-II-ID-PCE-2011-3-0256

\section{Metric Compatible Finsler Geometries}
\label{s2} In this section, we provide an introduction and analyze some
common features and differences of (pseudo) Riemannian and metric compatible
Finsler geometry models (proofs are omitted, see details in Refs. \cite%
{vrev1,vcrit,vsgg}).  In section \ref{s3}, we shall analyze the most important
formulas for metric compatible and noncompatible Finsler geometry models.
We emphasize that in Ricci flow theories, it is convenient to work both with global and coordinate/index formulas and equations. Some historical remarks will be presented in order to explain the most important ideas and results in Finsler geometry and related evolution/gravity theories.

\subsection{Finsler and Riemannian metrics}

Let $M$ be a real $\mathit{C}^{\infty }$ manifold, $\dim M=n$, and denote by $TM$ its tangent bundle. Denoting by $T_{x}M$ the tangent spaces at $x\in M,$ we have $TM=\bigcup\nolimits_{x\in M}T_{x}M.$

A Finsler fundamental/generating function (metric) is a function $F:\
TM\rightarrow \lbrack 0,\infty)$ subjected to the conditions:

\begin{enumerate}
\item $F(x,y)$ is $\mathit{C}^{\infty }$ on $\widetilde{TM}:=TM\backslash
\{0\},$ where $\{0\}$ denotes the set of zero sections of $TM$ on $M;$

\item $F(x,\beta y)=\beta F(x,y),$ for any $\beta >0,$ i.e. it is a positive
1--homogeneous function on the fibers of $TM;$

\item $\forall y\in \widetilde{T_{x}M},$ the Hessian $\ ^{v}{\tilde{g}}_{ij}$
(\ref{hessian}) is nondegenerate and positive definite\footnote{%
this condition should be relaxed for models of Finsler gravity with finite,
in general, locally anisotropic speed of light \cite{vrev1,vcrit}}.
\end{enumerate}

The term ''metric'' for $F$ is used in Finsler geometry because it defines
on $TM$ a nonlinear quadratic element%
\begin{equation}
ds^{2}=F^{2}(x,dx)  \label{nqe}
\end{equation}%
for $dx^{i}\sim y^{i}.$ The well--known and very important example of
(pseudo) Riemannian geometry, determined by a metric tensor $g_{ij}(x^{k})$,
is a particular case with quadratic form $F=\sqrt{|g_{ij}(x)y^{i}y^{j}|}$
when
\begin{equation}
ds^{2}=g_{ij}(x)dx^{i}dx^{j}  \label{lqe}
\end{equation}%
and the signature of $g_{ij}$ is of type $(+,+,+,+),$ or $(+,+,+,-),$ for
corresponding space like, or spacetime, manifolds. It should be noted that
the condition (\ref{lqe}) allows us to identify the fiber of $TM$ with a
flat (pseudo) Euclidean space, respectively, Minkowski spacetime, in any
point $x\in M.$ The tangent spaces $T_{x}M$ are considered in (pseudo)
Riemannian geometry on $M$ in order to define geometrically tensors and
forms by analogy to flat spaces. For instance, a vector $A=\{A_{i}(x)\}\in
TM $ in any system of reference/coordinates, has coefficients $A_{i}(x)$
depending only on base coordinates $x^{k}$ but not on $y^{a}.$ The
fundamental geometric objects (for instance, the Levi--Civita connection $%
\nabla $ and respective curvature tensor and Ricci tensor) are completely
and uniquely determined by a metric tensor $\ ^{h}g=\{g_{ij}(x)\}$ following
the condition of metric compatibility and zero torsion. This is a result of the ''quadratic''
condition (\ref{lqe}) when, in general, geometric and/or gravity theory
models based on (pseudo) Riemannian geometry, and various
Einstein/Riemann--Cartan or metric--affine generalizations, are for
geometric/physical objects depending only on $x$--coordinates. Any given
(pseudo) Riemannian metric structure naturally generates a unique chain
  $\ ^{h}g(x)\rightarrow \nabla (x)\rightarrow \ ^{\nabla }\mathcal{R}%
(x)\rightarrow \ ^{\nabla }Ric(x)$
following well--defined geometric rules. The "standard" theory of Ricci flows %
\cite{ham1,ham2,gper1,gper2,gper3} was formulated for (pseudo) Riemannian)
metrics $g_{ij}(\chi, x)$ depending on a real flow parameter $\chi $ (for
simplicity, we omit details on geometric flows of (almost) K\"{a}hler
geometries).

If a Finsler metric $F$ is generic nonlinear, the problem of constructing
geometric models on $TM$ became more sophisticate. Any relation of type (\ref%
{nqe}) for a class of correspondingly defined functions $F$ allows us to
study various metric properties of $T_{x}M$ and, in general, of $TM,$
including fiber constructions, using $\ ^{v}{\tilde{g}}_{ij}(x,y)$ (\ref%
{hessian}) and its possible projections, conformal transforms etc. For
instance, it is well known that B. Riemann in his famous thesis \cite%
{riemann} considered the first example of Finsler metric with nonlinear
quadratic elements (see historical remarks and references in \cite%
{cartan,matsumoto,ma,bejancu,bcs}; that why the term Riemann--Finsler
geometry was introduced in modern literature) even he elaborated a complete
geometric model only for Riemannian spaces. Nevertheless, to know the metric
properties is not enough for constructing a complete geometric model on $TM$
for a given $F$ and respective $\ ^{v}{\tilde{g}}_{ij}.$ We need more
assumptions, for instance, how we chose to define connections naturally
determined by $F$ because for generic Finsler metrics there is not a unique
analog of the Levi--Civita connection.

\subsection{Cartan--Finsler geometry}

The first complete geometric model of Finsler geometry is due to E. Cartan %
\cite{cartan}. Roughly speaking, the Cartan--Finsler geometry is a
variant of the well known Riemann--Cartan one, with nonzero torsion,
but constructed on $TM$ in a form when all geometric objects are
generated by $F$ following the conditions of metric compatibility
and vanishing of "pure" horizontal and vertical components of
torsion. Here we note that the Cartan--Finsler
torsion  is different from that used, for instance, in Einstein--Cartan
gravity when torsion is considered as an additional (to metric) tensor
field for which additional (algebraic) field equations are introduced. For the
Cartan--Finsler model, the torsion field is completely determined by metrics
$F$ and $\ ^{v}{\tilde{g}}_{ij}[F],$ when (at least, in principle) a complete
metric $\ ^{F}\mathbf{g}$ can be constructed on total  $TM$ following certain well defined geometric principles.

\subsubsection{The canonical N--connection, adapted frames and metrics}

Nevertheless, the Cartan--Finsler space is not only a Riemann--Cartan
geometry on $TM$ with metric tensor and metric compatible connection with
torsion (all induced by $F$). This is also an example of nonholonomic
manifold/bundle space when the geometric objects are adapted to a
non--integrable distribution on $TM$ induced by $F$ in such a form that
canonical semi--spray and nonlinear connection (N--connection) structures
are defined. In the mentioned first monograph on Finsler geometry \cite%
{cartan}, the concept of N--connection is considered in coordinate form (the
first global definitions are due to Ehresmann \cite{ehresmann} and A.
Kawaguchi \cite{kawaguchi1,kawaguchi2}, see details in \cite{ma} and, for
the Einstein gravity and generalizations, in \cite{vrev1,vsgg}). Let us
analyze, in brief, such constructions. A N--connection $\mathbf{N}$ can be
defined as a non--integrable (there are used equivalent terms like
nonholonomic and/or anholonomic) distribution
\begin{equation}
TTM=hTM\oplus vTM  \label{whitneyt}
\end{equation}%
into conventional horizontal (h) and vertical (v) subspaces\footnote{%
In our works, we use ''boldface'' symbols for spaces and geometric objects
endowed/adapted to N--connection structure.}. Locally, such a geometric
object is determined by its coefficients $\{N_{i}^{a}\}$, when $\mathbf{N=}%
N_{i}^{a}(u)dx^{i}\otimes \partial /\partial y^{a}$, and characterized by
its curvature (Neijenhuis tensor)
  $\mathbf{\Omega }=\frac{1}{2}\Omega _{ij}^{a}\ d^{i}\wedge d^{j}\otimes
\partial _{a},$ %
with coefficients%
\begin{equation}
\Omega _{ij}^{a}=\frac{\partial N_{i}^{a}}{\partial x^{j}}-\frac{\partial
N_{j}^{a}}{\partial x^{i}}+N_{i}^{b}\frac{\partial N_{j}^{a}}{\partial y^{b}}%
-N_{j}^{b}\frac{\partial N_{i}^{a}}{\partial y^{b}}.  \label{ncurv}
\end{equation}

In Cartan--Finsler geometry, the N--connection is canonically determined by $%
F $ following a geometric/variational principle: The value $L=F^{2}$ is
considered as an effective regular Lagrangian on $TM$ and action integral
$$S(\tau )=\int\limits_{0}^{1}L(x(\tau ),y(\tau ))d\tau ,\mbox{ for }%
y^{k}(\tau )=dx^{k}(\tau )/d\tau ,$$%
for $x(\tau )$ parametrizing smooth curves on a manifold $M$ with $\tau \in
\lbrack 0,1].$ The Euler--Lagrange equations $\frac{d}{d\tau }\frac{\partial
L}{\partial y^{i}}-\frac{\partial L}{\partial x^{i}}=0$ are equivalent to
the ''nonlinear geodesic'' (equivalently, semi--spray) equations $\frac{%
d^{2}x^{k}}{d\tau ^{2}}+2\tilde{G}^{k}(x,y)=0$, where
\begin{equation}
\tilde{G}^{k}=\frac{1}{4}\tilde{g}^{kj}\left( y^{i}\frac{\partial ^{2}L}{%
\partial y^{j}\partial x^{i}}-\frac{\partial L}{\partial x^{j}}\right) ,
\label{smspr}
\end{equation}%
for $\tilde{g}^{kj}$ being inverse to $\ ^{v}{\tilde{g}}_{ij}\equiv {\tilde{g%
}}_{ij}$ (\ref{hessian}), defines the canonical N--connection $\ $%
\begin{equation}
\tilde{N}_{j}^{a}:=\frac{\partial \tilde{G}^{a}(x,y)}{\partial y^{j}}.
\label{cncl}
\end{equation}

A fundamental Finsler function $F(x,y)$ induces naturally a N--adapted frame
structure (defined linearly by $\ \tilde{N}_{j}^{a}$), $\mathbf{\tilde{e}}%
_{\nu }=(\mathbf{\tilde{e}}_{i},e_{a}),$ where
\begin{equation}
\mathbf{\tilde{e}}_{i}=\frac{\partial }{\partial x^{i}}-\tilde{N}_{i}^{a}(u)%
\frac{\partial }{\partial y^{a}}\mbox{ and
}e_{a}=\frac{\partial }{\partial y^{a}},  \label{dder}
\end{equation}%
and the dual frame (coframe) structure is $\mathbf{\tilde{e}}^{\mu }=(e^{i},%
\mathbf{\tilde{e}}^{a}),$ where
\begin{equation}
e^{i}=dx^{i}\mbox{ and }\mathbf{e}^{a}=dy^{a}+\tilde{N}_{i}^{a}(u)dx^{i}.
\label{ddif}
\end{equation}%
There are satisfied nontrivial nonholonomy relations
\begin{equation}
\lbrack \mathbf{\tilde{e}}_{\alpha },\mathbf{\tilde{e}}_{\beta }]=\mathbf{%
\tilde{e}}_{\alpha }\mathbf{\tilde{e}}_{\beta }-\mathbf{\tilde{e}}_{\beta }%
\mathbf{\tilde{e}}_{\alpha }=\tilde{W}_{\alpha \beta }^{\gamma }\mathbf{%
\tilde{e}}_{\gamma }  \label{anhrel}
\end{equation}%
with (antisymmetric) nontrivial anholonomy coefficients $\tilde{W}%
_{ia}^{b}=\partial _{a}\tilde{N}_{i}^{b}$ and $\tilde{W}_{ji}^{a}=\tilde{%
\Omega}_{ij}^{a}.$ This is a reason to say that a Finsler geometry is a
nonholonomic one when $F$ defines a ''preferred'' frame structure on $TM$.%
\footnote{%
Such a N--adapted frame system of reference does not prohibits us to
consider arbitrary frame and coordinate transforms on $TM$.} If a generating
function $F$ is of particular quadratic type (\ref{lqe}), the values $%
\tilde{N}_{j}^{a},$ $\mathbf{\tilde{e}}_{\alpha }$ and $\tilde{W}_{\alpha
\beta }^{\gamma }$ can be parametrized in some forms not depending
explicitly on $y^{a}.$ In such cases, $\mathbf{\tilde{e}}_{\alpha }$ can be
arbitrary frames not depending on a ''degenerate'' Finsler, i.e. on a
(pseudo) Riemannian metric $g_{ij}(x).$

Using data $\left( {\tilde{g}}_{ij},\mathbf{\tilde{e}}_{\alpha }\right) ,$
we can define a canonical (Sasaki type) metric structure on $\widetilde{TM},$%
\begin{equation}
\mathbf{\tilde{g}}=\tilde{g}_{ij}(x,y)\ e^{i}\otimes e^{j}+\tilde{g}%
_{ij}(x,y)\ \mathbf{\tilde{e}}^{i}\otimes \ \mathbf{\tilde{e}}^{j}.
\label{slm}
\end{equation}%
It is possible to use other geometric principles for ''lifts and
projections'' with $\ ^{v}{\tilde{g}}_{ij}$ on the typical fiber, when from
a given $F$ it is constructed a metric on total/horizontal spaces of $TM$.
Nevertheless, for models of locally anisotropic/Finsler gravity on $TM,$
with a generalized covariance principle, such constructions are equivalent
up to certain frame/coordinate transforms $\mathbf{\tilde{e}}_{\gamma
}\rightarrow \mathbf{e}_{\gamma ^{\prime }}=e_{\ \gamma ^{\prime }}^{\gamma }%
\mathbf{\tilde{e}}_{\gamma }.$ In such cases, we can omit ''tilde'' on
symbols and write, in general, $\mathbf{g=\{g}_{\alpha \beta }\mathbf{\}}$
and $\mathbf{N=\{}N_{i}^{a}=e_{\ a^{\prime }}^{a}e_{i}^{\ i^{\prime
}}N_{i^{\prime }}^{a^{\prime }}\}.$ There is a subclass of transforms
preserving a prescribed splitting (\ref{whitneyt}).

We note that in Finsler geometry and generalizations there are used terms
like distinguished tensor/ metric/ connection etc (in brief, d--tensor,
d--metric, d--connection) for geometric objects adapted to N--connec\-ti\-on
splitting when coefficients are computed with respect to frames of type (\ref%
{dder}) and (\ref{ddif}). For instance, a d--vector $\mathbf{X=(\ }^{h}X,%
\mathbf{\ }^{v}X\mathbf{)=}X^{i}\mathbf{\tilde{e}}_{i}+X^{a}e_{a}.$

\subsubsection{Torsion and curvature of d--connections}

For any d--metric $\mathbf{\tilde{g}}$ (\ref{slm}), we may construct in
standard form, on $TM$, its Levi--Civita connection $\tilde{\nabla}.$
Nevertheless, such a linear connection is not used in Finsler geometry
because it is not adapted to the N--connection structure $\mathbf{N.}$ This
motivates the definition of a new class of linear connections.

A distinguished connection (d--connection) $\mathbf{D}$ on $TM$ is a linear
connection conserving under parallelism the Whitney sum (\ref{whitneyt}).
For any $\mathbf{D}$, there is a decomposition into h-- and v--covariant
derivatives,
$$\mathbf{D}_{\mathbf{X}}\mathbf{\doteqdot X}\rfloor \mathbf{D=\ }^{h}X\rfloor
\mathbf{D+}\ \mathbf{\ }^{v}X\rfloor \mathbf{D=}D_{\mathbf{\ }^{h}X}+D_{%
\mathbf{\ }^{v}X}=\mathbf{\ }^{h}D_{X}+\mathbf{\ }^{v}D_{X},$$
where ''$\rfloor "$ denotes the interior product.

The torsion of a d--connection $\mathbf{D}$ is defined in standard from by
d--tensor field
\begin{equation}
\mathcal{T}(\mathbf{X},\mathbf{Y}):= \mathbf{D}_\mathbf{X}\mathbf{Y}-
\mathbf{D}_\mathbf{Y}\mathbf{X}-[\mathbf{X},\mathbf{Y}],  \label{tors1}
\end{equation}%
for which a N--adapted $h$-$v$--decomposition is possible, $\mathcal{T}(\mathbf{X},\mathbf{Y})=T(\ ^hX,\ ^hY)+ T(\ ^hX, ^vY)+T(\ ^vX,\
^hY)+T(\ ^vX,\ ^vY).$
The curvature of a d--connection $\mathbf{D}$ is
\begin{equation}
\mathcal{R}(\mathbf{X},\mathbf{Y}):= \mathbf{D}_\mathbf{X}\mathbf{D}_\mathbf{%
Y}- \mathbf{D}_\mathbf{Y}\mathbf{D}_\mathbf{X}-\mathbf{D}_\mathbf{[X,Y]}
\label{curv1}
\end{equation}%
for any d--vectors $\mathbf{X,Y,}$ with a corresponding $h$--$v$%
--decomposition (for simplicity, we omit such formulas) .

The N--adapted components $\mathbf{\Gamma }_{\ \beta \gamma }^{\alpha }$ of
a d--connection $\mathbf{D}_{\alpha }=(\mathbf{e}_{\alpha }\rfloor \mathbf{D}%
)$ are computed following equations $\mathbf{D}_{\alpha }\mathbf{e}_{\beta }=%
\mathbf{\Gamma }_{\ \alpha \beta }^{\gamma }\mathbf{e}_{\gamma }$, or $%
\mathbf{\Gamma }_{\ \alpha \beta }^{\gamma }\left( u\right) =\left( \mathbf{D%
}_{\alpha }\mathbf{e}_{\beta }\right) \rfloor \mathbf{e}^{\gamma }$.
Respective splitting into h-- and v--covariant derivatives are given by
$h\mathbf{D}=\{\mathbf{D}_{k}=\left( L_{jk}^{i},L_{bk\;}^{a}\right) \}$
 and $\ v\mathbf{D}=\{\mathbf{D}_{c}=\left( C_{jc}^{i},C_{bc}^{a}\right) \}$
where,
 $L_{jk}^{i}=\left( \mathbf{D}_{k}\mathbf{e}_{j}\right) \rfloor e^{i},$ $
L_{bk}^{a}=\left( \mathbf{D}_{k}e_{b}\right) \rfloor \mathbf{e}%
^{a},~C_{jc}^{i}=\left( \mathbf{D}_{c}\mathbf{e}_{j}\right) \rfloor
e^{i},\quad C_{bc}^{a}=\left( \mathbf{D}_{c}e_{b}\right) \rfloor \mathbf{e}%
^{a}.$
A set \\  $\mathbf{\Gamma }_{\ \alpha \beta }^{\gamma }=\left(
L_{jk}^{i},L_{bk}^{a},C_{jc}^{i},C_{bc}^{a}\right) $ completely define a
d--connection $\mathbf{D}$ on $TM$ enabled with N--connection structure.

The simplest way to perform computations with d--connections is to use
N--adapted differential forms. The d--connection 1--form is $\mathbf{\Gamma }%
_{\ \beta }^{\alpha }=\mathbf{\Gamma }_{\ \beta \gamma }^{\alpha }\mathbf{e}%
^{\gamma }$. For instance, the $h$--$v$--coefficients $\mathbf{T}_{\ \beta
\gamma }^{\alpha }=\{T_{\ jk}^{i},T_{\ ja}^{i},T_{\ ji}^{a},T_{\
bi}^{a},T_{\ bc}^{a}\}$ of torsion $\mathbf{T}$ (\ref{tors1}) are computed
using formulas $
\mathcal{T}^{\alpha }:=\mathbf{De}^{\alpha }=d\mathbf{e}^{\alpha }+\mathbf{%
\Gamma }_{\ \beta }^{\alpha }\wedge \mathbf{e}^{\beta }.$
We obtain
{\small
\begin{eqnarray}
T_{\ jk}^{i} &=& L_{\ jk}^{i}-L_{\ kj}^{i},\ T_{\ ja}^{i}=-T_{\ aj}^{i}=C_{\
ja}^{i},\ T_{\ ji}^{a}=\Omega _{\ ji}^{a}, \nonumber \\
T_{\ bi}^{a} &=& \frac{\partial N_{i}^{a}}{\partial y^{b}}-L_{\ bi}^{a},\
T_{\ bc}^{a}=C_{\ bc}^{a}-C_{\ cb}^{a}.  \label{dtors}
\end{eqnarray}%
}
Similarly, we can compute the N--adapted components $\mathbf{R}_{\ \beta
\gamma \delta }^{\alpha }$ of  curvature  (\ref{curv1}), \begin{equation}
\mathcal{R}_{~\beta }^{\alpha }\doteqdot \mathbf{D\Gamma }_{\ \beta
}^{\alpha }=d\mathbf{\Gamma }_{\ \beta }^{\alpha }-\mathbf{\Gamma }_{\ \beta
}^{\gamma }\wedge \mathbf{\Gamma }_{\ \gamma }^{\alpha }=\mathbf{R}_{\ \beta
\gamma \delta }^{\alpha }\mathbf{e}^{\gamma }\wedge \mathbf{e}^{\delta }.
\label{curv1a}
\end{equation}%
For simplicity, we omit formulas for an explicit $h$--$v$--parametrization
of $\mathbf{R}_{\ \beta \gamma \delta }^{\alpha },$ see details in Refs. %
\cite{vrev1,vsgg,ma}.

There is a sub--class of d--connections $\mathbf{D}$ on $\mathbf{TM}$ which
are metric compatible to a d--metric
\begin{equation}
\mathbf{g}=\ g_{ij}(x,y)\ e^{i}\otimes e^{j}+\ h_{ab}(x,y)\ \mathbf{e}%
^{a}\otimes \mathbf{e}^{b}  \label{m1}
\end{equation}%
with  N--adapted decomposition $\mathbf{g=}hg\mathbf{\oplus _{N}}%
vg=[hg,vg]. $\footnote{%
Any d--metric $\mathbf{g}=g_{\alpha \beta }du^{\alpha }du^{\beta }$ on $TM,$
via corresponding frame/coordinate transforms can be parametrized in the
form (\ref{m1}) and $\mathbf{\tilde{g}}$ (\ref{slm}) (in the last case, we
have to prescribe a generating function $F).$ This mean that on the total
space of a tangent bundle endowed with metric structure $\mathbf{g}$ we can
always introduce Finsler variables when $\mathbf{g=\tilde{g}}$ and there is
a $h$--$v$--splitting $\mathbf{N=\tilde{N}.}$ The constructions are
performed equivalently but depend on the type of geometric structure chosen
to be the fundamental one. If $F$ is prescribed, then we construct data $%
\left( F:\mathbf{\tilde{N},\tilde{g}}\right) $ which up to frame transforms [%
$\mathbf{e}_{\gamma ^{\prime }}=e_{\ \gamma ^{\prime }}^{\gamma }\mathbf{%
\tilde{e}}_{\gamma };$ the vielbeins $e_{\ \gamma ^{\prime }}^{\gamma }$
have to be defined as a solution of an algebraic quadratic equations $%
\mathbf{g}_{\alpha ^{\prime }\beta ^{\prime }}=e_{\ \alpha ^{\prime
}}^{\alpha }e_{\ \beta ^{\prime }}^{\beta }\mathbf{\tilde{g}}_{\alpha \beta
} $ for given $\mathbf{g}_{\alpha ^{\prime }\beta ^{\prime }}$ and $\mathbf{%
\tilde{g}}_{\alpha \beta }$] \ are equivalent to some data $(\mathbf{N,g}).$
Inversely, we can fix any $(\mathbf{N,g})$ (in particular, $\mathbf{N}$ can
be for any conventional $h$--$v$--splitting) and then chose any convenient $%
F $ when via frame transforms $(\mathbf{N,g})\rightarrow (\mathbf{\tilde{N},%
\tilde{g}}).$} The condition of compatibility $\mathbf{Dg=0}$ split in
respective conditions for $h$-$v$--components, $%
D_{j}g_{kl}=0,D_{a}g_{kl}=0,D_{j}h_{ab}=0,D_{a}h_{bc}=0.$

We can construct a canonical d--connection $\widehat{\mathbf{D}}$ completely
defined by a d--metric $\mathbf{g}$ (\ref{m1}) in metric compatible form, $%
\widehat{\mathbf{D}}\mathbf{g=}0,$ and with zero $h$- and $v$-torsions (with
$\widehat{T}_{\ jk}^{i}=0$ and $\widehat{T}_{\ bc}^{a}=0$ but, in general,
nonzero $\widehat{T}_{\ ja}^{i},\widehat{T}_{\ ji}^{a}$ and $\widehat{T}_{\
bi}^{a},$ see (\ref{dtors})). The coefficients of $\widehat{\mathbf{D}},$
computed with respect to N--adapted frames are $\widehat{\mathbf{\Gamma }}%
_{\ \alpha \beta }^{\gamma }=\left( \widehat{L}_{jk}^{i},\widehat{L}%
_{bk}^{a},\widehat{C}_{jc}^{i},\widehat{C}_{bc}^{a}\right) $ for
\begin{eqnarray}
\widehat{L}_{jk}^{i} &=&\frac{1}{2}g^{ir}\left( \mathbf{e}_{k}g_{jr}+\mathbf{%
e}_{j}g_{kr}-\mathbf{e}_{r}g_{jk}\right), \nonumber\\
\widehat{L}_{bk}^{a} &=& e_{b}(N_{k}^{a})+\frac{1}{2}h^{ac}\left(
e_{k}h_{bc}-h_{dc}\ e_{b}N_{k}^{d}-h_{db}\ e_{c}N_{k}^{d}\right) ,   \label{candcon}  \\
\widehat{C}_{jc}^{i} &=&\frac{1}{2}g^{ik}e_{c}g_{jk},\ \widehat{C}_{bc}^{a}=%
\frac{1}{2}h^{ad}\left( e_{c}h_{bd}+e_{c}h_{cd}-e_{d}h_{bc}\right) .  \nonumber
\end{eqnarray}

For any metric structure $\mathbf{g}$ on $TM$, we can compute also the
Levi--Civita connection $\nabla $ for which $\ ^{\nabla }\mathcal{T}^{\alpha
}=0$ and $\nabla \mathbf{g=0.}$ There is a canonical distortion relation
\begin{equation}
\widehat{\mathbf{D}}\mathbf{=}\nabla +\widehat{\mathbf{Z}}  \label{dcdc}
\end{equation}%
where both connections $\widehat{\mathbf{D}},\nabla $ and $\widehat{\mathbf{Z%
}}$ (such a distortion tensor is an algebraic combination of nontrivial
torsion coefficients $\widehat{T}_{\ ja}^{i},\widehat{T}_{\ ji}^{a}$ and $%
\widehat{T}_{\ bi}^{a}$) are uniquely defined by the same metric structure $%
\mathbf{g}$. Taking $\mathbf{g=\tilde{g},}$ such values $\widehat{\mathbf{%
\tilde{D}}},\tilde{\nabla}$ and $\widehat{\mathbf{\tilde{Z}}}$ can be
derived from a Finsler metric $F$ (for simplicity, we omit explicit
coordinate formulas for $\nabla $ and $\widehat{\mathbf{Z}},$ see details in %
\cite{vrev1,vsgg,ma}). This allows us to construct a complete model of
Finsler space on $TM.$ Such a canonical metric compatible geometry is
determined by data $\left( F:\mathbf{g,N,}\widehat{\mathbf{D}}\right).$

\subsubsection{The Cartan d--connection}

Historically, E. Cartan \cite{cartan} used another type of metric compatible
d--connection $\mathbf{\tilde{D}}$ which via frame transforms and
deformations can be related to $\widehat{\mathbf{D}}$ (\ref{candcon}). If we
consider that $\widehat{L}_{bk}^{a}\rightarrow \widehat{L}_{jk}^{i}$ and $%
\widehat{C}_{jc}^{i}\rightarrow $ $\ \widehat{C}_{bc}^{a},$ by identifying
respectively $a=n+i$ with $i$ and $b=n+j,$ we obtain the so--called normal
d--connection \ $\ ^{n}\mathbf{D}=(\widehat{L}_{\ jk}^{i},\widehat{C}%
_{jc}^{i})$ with N--adapted 1--form
 $\widehat{\mathbf{\Gamma }}_{\ j}^{i}=\widehat{\mathbf{\Gamma }}_{\ j\gamma
}^{i}\mathbf{e}^{\gamma }=\widehat{L}_{\ jk}^{i}e^{k}+\widehat{C}_{jc}^{i}%
\mathbf{e}^{c},$
where
\begin{equation}
\widehat{L}_{\ jk}^{i}=\frac{1}{2}g^{ih}(\mathbf{e}_{k}g_{jh}+\mathbf{e}%
_{j}g_{kh}-\mathbf{e}_{h}g_{jk}),\widehat{C}_{\ bc}^{a}=\frac{1}{2}%
g^{ae}(e_{b}h_{ec}+e_{c}h_{eb}-e_{e}h_{bc}).  \label{cdc}
\end{equation}%
Taking $\mathbf{g=\tilde{g},}$ when $\tilde{h}_{ij}=\tilde{g}_{ij},$ \ and $%
\mathbf{N=\tilde{N}}$ in (\ref{cdc}), we define the Cartan d--connection \ $%
\mathbf{\tilde{D}}=(\tilde{L}_{\ jk}^{i},\tilde{C}_{jc}^{i}).$

For $\mathbf{\tilde{D},}$ the nontrivial h-- and v--components of torsion $%
\mathbf{\tilde{T}}_{\beta \gamma }^{\alpha }=\{\tilde{T}_{jc}^{i},\tilde{T}%
_{ij}^{a},\tilde{T}_{ib}^{a}\}$ and curvature $\mathbf{\tilde{R}}_{\ \beta
\gamma \tau }^{\alpha }=\{\tilde{R}_{\ hjk}^{i},\tilde{P}_{\ jka}^{i},\tilde{%
S}_{\ bcd}^{a}\}$ are respectively
\begin{eqnarray}
\tilde{T}_{jc}^{i}&=&\tilde{C}_{\ jc}^{i},\tilde{T}_{ij}^{a}=\tilde{\Omega}%
_{ij}^{a},\tilde{T}_{ib}^{a}=e_{b}\left( \tilde{N}_{i}^{a}\right) -\tilde{L}%
_{\ bi}^{a},  \label{torscdc}\\
\tilde{R}_{\ hjk}^{i} &=&\mathbf{\tilde{e}}_{k}\tilde{L}_{\ hj}^{i}-\mathbf{%
\tilde{e}}_{j}\tilde{L}_{\ hk}^{i}+\tilde{L}_{\ hj}^{m}\tilde{L}_{\ mk}^{i}-%
\tilde{L}_{\ hk}^{m}\tilde{L}_{\ mj}^{i}-\tilde{C}_{\ ha}^{i}\tilde{\Omega}%
_{\ kj}^{a},  \label{curvcart} \\
\tilde{P}_{\ jka}^{i} &=&e_{a}\tilde{L}_{\ jk}^{i}-\mathbf{\tilde{D}}_{k}%
\tilde{C}_{\ ja}^{i},\ \tilde{S}_{\ bcd}^{a}=e_{d}\tilde{C}_{\ bc}^{a}-e_{c}%
\tilde{C}_{\ bd}^{a}+\tilde{C}_{\ bc}^{e}\tilde{C}_{\ ed}^{a}-\tilde{C}_{\
bd}^{e}\tilde{C}_{\ ec}^{a}.  \nonumber
\end{eqnarray}%
We note that $h$-- and $v$--components of torsion are zero, $\tilde{T}%
_{jk}^{i}=0$ and $\tilde{T}_{bc}^{a}=0,$ even there are also nontrivial
components $\tilde{T}_{ij}^{a}$ and $\tilde{T}_{ib}^{a}.$

The Cartan d--connection is characterized by a unique distortion relation
\begin{equation}
\mathbf{\tilde{D}=}\tilde{\nabla}+\mathbf{\tilde{Z}},  \label{dcartdc}
\end{equation}%
where all values $\mathbf{\tilde{D},}\tilde{\nabla}$ and $\mathbf{\tilde{Z}}$
are determined (up to frame/coordinate transforms) by $F$ and $\mathbf{%
\tilde{g}.}$ On $TM,$ the data $\left( F;\mathbf{\tilde{g},\tilde{N},\tilde{D%
}}\right) $ define a model of Cartan--Finsler geometry.

\subsubsection{The almost K\"{a}hler model of Cartan--Finsler geometry}

\label{akmfg}There is a fundamental result by M. Matsumoto \cite{matsumoto}
which allows us to reformulate $\left( F;\mathbf{\tilde{g},\tilde{N},\tilde{D%
}}\right) ,$ equivalently, as an almost K\"{a}hler geometry. Let us consider
 a linear operator $\mathbf{\tilde{J}}$ acting on vectors on $TM$
following formulas
 $
\mathbf{\tilde{J}}(\mathbf{\tilde{e}}_{i})=-e_{i}$  and $\mathbf{%
\tilde{J}}(e_{i})=\mathbf{\tilde{e}}_{i},$
where the superposition $\mathbf{\tilde{J}\circ \tilde{J}=-I,}$ for the
unity matrix $\mathbf{I.}$

A Finsler fundamental function $F(x,y)$ and the corresponding Sasaki type
metric $\mathbf{\tilde{g}}$ (\ref{slm}) induce, respectively, a canonical
1--form $\tilde{\omega}=F\frac{\partial F}{\partial y^{i}}e^{i}$ and a
canonical 2--form
\begin{equation}
\mathbf{\tilde{\theta}}=\ \tilde{g}_{ij}(x,y)\mathbf{e}^{i}\wedge e^{j}.
\label{asstr}
\end{equation}%
Such objects are associated to $\mathbf{J}$ following formula $\ \mathbf{%
\tilde{\theta}(X,Y)}:= \mathbf{\tilde{g}}(\mathbf{\tilde{J}X,Y})$ for any
d--vectors $\mathbf{X}$ and $\mathbf{Y.}$ By straightforward computations,
we can prove that $d\tilde{\omega}=\mathbf{\tilde{\theta}}$. This states on $%
TM$ an almost Hermitian (symplectic) structure nonholonomically induced by $%
F $. Considering $\ ^{\theta }\mathbf{D\equiv \tilde{D}}$ as an almost
symplectic d--connection, we can prove that
  $\ ^{\theta }\mathbf{D}_{\mathbf{X}}\mathbf{\tilde{\theta}=0}$ and $\
^{\theta }\mathbf{D}_{\mathbf{X}}\mathbf{\tilde{J}=0.}$
  The data $(F;\mathbf{\tilde{\theta},\tilde{J},}\ ^{\theta }\mathbf{D})$
define a nonholonomic almost K\"{a}hler space.

It should be noted that canonical almost symplectic/K\"{a}hler variables $%
\tilde{\theta},\tilde{\mathbf{J}},$ and $\ ^{\theta }\mathbf{D}$ can be
introduced for any $TM$ endowed with d--metric, $\mathbf{g}$, and
N--connection, $\mathbf{N,}$ structures. For this, we have to prescribe an
effective generating function $F$ and compute $\mathbf{\tilde{e}}_{\alpha }$
and $\mathbf{\tilde{g}}$. Solving a quadratic algebraic equation to
construct $\mathbf{\tilde{e}} _{\gamma }\rightarrow \mathbf{e}_{\gamma
^{\prime }}= e_{\ \gamma ^{\prime }}^{\gamma }\mathbf{\tilde{e}}_{\gamma },$
we encode equivalently and data $(TM,\mathbf{g})$ as a Cartan--Finsler
model, $(F; \mathbf{\tilde{g},\tilde{N},\tilde{D}}),$ and/or an almost K\"{a}%
hler--Finsler model, $(F; \mathbf{\tilde{\theta},\tilde{J},}\ ^{\theta}%
\mathbf{D})$. Such results were used for deformation quantization of
Lagrange--Finsler spaces \cite{vdefqfl}. Finally, we cite an alternative
approach with K\"{a}hler structures associated to Berwald or Randers metrics
etc \cite{peyghan1,peyghan3,peyghan4}.

\section{Metric Noncompatible Finsler Spaces}

\label{s3} There were developed alternative approaches to constructing
geometric models determined by a fundamental Finsler function $F(x,y).$ In
some sense, mathematicians attempted to formulate a more ''simple'' version
of Finsler geometry than the Cartan model, not mimicking on tangent bundles
a variant of nonholonomic Riemann space. Chronologically, the first metric
noncompatible models were proposed by L. Berwald \cite{berwald} and S. Chern %
\cite{chern} (see details in \cite{bcs}). More recently, a different
"nonstandard" construction for the Ricci curvature was proposed by H.
Akbar--Zadeh \cite{akbar}. In this section, we outline three important
models of ''non--Cartan'' Finsler spaces.

\begin{itemize}
\item The \textbf{Berwald d--connection } is $\ ^{B}\mathbf{D:}=(\ ^{B}L_{\
jk}^{i}=\partial \tilde{N}_{j}^{i}/\partial y^{k},\ ^{B}C_{jc}^{i}=0)$, when
the h--covariant derivative is defined by the first v--derivatives of
Cartan's N--connection structure $\tilde{N}$ (\ref{cncl}) and an additional
constraint that the v--covariant derivative is zero is imposed.

\item The \textbf{Chern d--connection} is $\ ^{Ch}\mathbf{D:}=(\ ^{Ch}L_{\
jk}^{i}=\tilde{L}_{\ jk}^{i},\ ^{Ch}C_{jc}^{i}=0)$, when the h--covariant
derivative is the Cartan's one computed as in the first formula for the
normal d--connection (\ref{cdc}), with an additional constraint that the
v--covariant derivative is zero is imposed.
\end{itemize}

Both geometries with $\ ^{B}\mathbf{D}$ and/or $\ ^{Ch}\mathbf{D}$ can be
modelled on $hTM.$ The Chern's d--connection keeps all properties of the
Levi--Civita connection for geometric constructions on the $h$--subspace.
Nevertheless, both d--connections are not metric compatible on total space
of $TM,$ i.e. there are nontrivial nonmetricity fields, $\mathcal{Q}:=%
\mathbf{Dg,}$$\ ^{B}\mathcal{Q}\neq 0$ and $\ ^{Ch}\mathcal{Q}\neq 0$. Such
nonmetricities, in general, present substantial difficulties in constructing
well--defined minimal Finsler extensions of the standard models of Finsler
gravity, see critical remarks in Refs. \cite{vcrit,vrev1,vsgg} (for
instance, there are problems with physical interpretation of nonmetricity
fields, definition of spinors and constructing Dirac operators, formulating
conservation laws etc).

Applying formulas (\ref{tors1}) and (\ref{curv1}), we compute respectively
the torsions $\ ^{B}\mathcal{T}\neq 0,\ ^{Ch}\mathcal{T}=0$ and curvatures $%
\ ^{B}\mathcal{R}\neq 0,\ ^{Ch}\mathcal{R}\neq 0$ as 2--forms.

In Refs. \cite{bao2,akbar}, it is used as curvature in Finsler geometry the
value $
\breve{\mathbf{R}}=\breve{R}_{k}^{i}\ dx^{k}\otimes \frac{\partial }{%
\partial x^{i}}|_{x}:T_{x}M\rightarrow T_{x}M,$ (this type of ''curvature'' is considered \ in a manner different that
definitions of curvatures with associated 2--forms) where
\begin{equation}
\breve{R}_{k}^{i}=2\frac{\partial \tilde{G}^{i}}{\partial x^{k}}-y^{j}\frac{%
\partial ^{2}\tilde{G}^{i}}{\partial x^{j}\partial y^{k}}+2\tilde{G}^{j}%
\frac{\partial ^{2}\tilde{G}^{i}}{\partial y^{j}\partial y^{k}}-\frac{%
\partial \tilde{G}^{i}}{\partial y^{j}}\frac{\partial \tilde{G}^{j}}{%
\partial y^{k}}  \label{flcurv}
\end{equation}%
is determined by semi--spray $\tilde{G}^{k}$ (\ref{smspr}). Such values are
convenient for study geometric objects in $T_{x}M$ for a point $x\in $ $M.$

\subsection{Nonholonomic deformations and distortions}

We note that above presented formulas for metric compatible and
noncommpatible d--connections in Finsler geometry are uniquely related via
certain distortion tensors of type (\ref{dcartdc}) and (\ref{dcdc}). In
order to derive deformations of fundamental tensor objects (for instance,
torsions and curvature) in N--adapted form it is convenient to perform all
constructions for the Cartan d--connection and then to compute distortions
for necessary tensors and differential forms. We can write%
\begin{equation}
\ ^{B}\mathbf{D}=\ \mathbf{\tilde{D}+}\ ^{B}\mathbf{\tilde{Z}}%
\mbox{\ and \ }\ ^{Ch}\mathbf{D}\mathbf{=}\ \mathbf{\tilde{D}+}\ ^{Ch}%
\mathbf{\tilde{Z},}  \label{distdb}
\end{equation}%
where all d--connections and distorting tensors are uniquely computed using
components of $\mathbf{\tilde{g}}$ and $\mathbf{\tilde{N}}$ for a chosen
fundamental Finsler function $F.$ Both d--connections are with nontrivial
nonmetricity $\ ^{B}\mathcal{Q}\neq 0$ and $\ ^{Ch}\mathcal{Q}\neq 0.$
Nevertheless, such tensor objects are not arbitrary ones but completely
induced by respective$\ ^{B}\mathbf{\tilde{Z}}$ and $\ ^{Ch}\mathbf{\tilde{Z}%
.}$

\subsubsection{Distortions of Ricci tensors for Finsler d--connections}

Hereafter, we shall denote by $\ ^{F}\mathbf{D}$ any d--connection (metric
compatible or not, for instance, of type (\ref{distdb})) uniquely determined
by $F.$ Computing the curvature 2--form (\ref{curv1a}) for $\ ^{F}\mathbf{D,}
$ $
\ ^{F}\mathcal{R}_{~\beta }^{\alpha }\doteqdot \ ^{F}\mathbf{D}\ ^{F}\mathbf{%
\Gamma }_{\ \beta }^{\alpha }=d\ ^{F}\mathbf{\Gamma }_{\ \beta }^{\alpha }-\
^{F}\mathbf{\Gamma }_{\ \beta }^{\gamma }\wedge \ ^{F}\mathbf{\Gamma }_{\
\gamma }^{\alpha }=\ ^{F}\mathbf{R}_{\ \beta \gamma \delta }^{\alpha }%
\mathbf{e}^{\gamma }\wedge \mathbf{e}^{\delta },$ %
we express
\begin{equation}
\ ^{F}\mathbf{R}_{\ \beta \gamma \delta }^{\alpha }=\ \mathbf{\tilde{R}}_{\
\beta \gamma \delta }^{\alpha }+\ \mathbf{\tilde{Z}}_{\ \beta \gamma \delta
}^{\alpha }.  \label{curvdist}
\end{equation}%
Contracting indices, $\ ^{F}\mathbf{R}_{\ \beta \gamma }:=\ ^{F}\mathbf{R}%
_{\ \beta \gamma \alpha }^{\alpha },$ we obtain the N--adapted coefficients
for the Ricci tensor $\ ^{F}\mathcal{R}\mathit{ic}\mathbf{\doteqdot \{}\ ^{F}%
\mathbf{\mathbf{R}}_{\beta \gamma }\mathbf{=}\left(
R_{ij},R_{ia},R_{ai},R_{ab}\right) \mathbf{\}.}$ This tensor, in general, is
not symmetric, $\ ^{F}\mathbf{\mathbf{R}}_{\beta \gamma }\mathbf{\neq }\ ^{F}%
\mathbf{\mathbf{R}}_{\gamma \beta }\mathbf{,}$ and corresponding Bianchi
identities result in constraints%
\begin{equation}
\ ^{F}\mathbf{D}^{\beta }\left( \ ^{F}\mathbf{\mathbf{R}}_{\beta \gamma }-%
\frac{1}{2}\mathbf{g}_{\beta \gamma }\ _{s}^{F}R\right) :=\ ^{F}\mathbf{J}%
_{\gamma }\neq 0  \label{nhconsl}
\end{equation}%
even for metric compatible $\ ^{F}\mathbf{D.}$ In above formulas, the scalar
curvature is by definition
\begin{equation}
\ _{s}^{F}R:=\mathbf{g}^{\beta \gamma }\ ^{F}\mathbf{\mathbf{R}}_{\beta
\gamma }=g^{ij}R_{ij}+h^{ab}R_{ab}.  \label{riccifs}
\end{equation}%
The Einstein tensor $\ ^{F}\mathbf{E}_{\beta \gamma }$ can be postulated in
standard form for any $\ ^{F}\mathbf{D,}$%
\begin{equation}
\ ^{F}\mathbf{E}_{\beta \gamma }:=\ ^{F}\mathbf{\mathbf{R}}_{\beta \gamma }-%
\frac{1}{2}\mathbf{g}_{\beta \gamma }\ _{s}^{F}R.  \label{einstfdt}
\end{equation}

The relations (\ref{nhconsl}) can be considered as a nonholonomic ''unique''
deformation of standard relations $\nabla _{\alpha }E^{\alpha \beta }=0,$
with Einstein tensor $E^{\alpha \beta }$ for the Levi--Civita connection $%
\nabla _{\alpha },$ in general relativity. Using distorting relations (\ref%
{distdb}), we can always compute the source $\ ^{F}\mathbf{J}_{\gamma }$ and
define an associated set of constraints as in nonholonomic mechanics.

The distortions of connections will be used in the nonholonomic geometric
flow theory  as follows. Contracting indices in (\ref{curvdist}), we
compute
{\small
\begin{eqnarray*}
\ ^{F}\mathbf{R}_{\ \beta \gamma }&=&\ \mathbf{\tilde{R}}_{\ \beta \gamma }+\
\mathbf{\tilde{Z}}_{\ \beta \gamma \delta }^{\alpha },  \label{districci}\\
\mathbf{\tilde{R}}_{\ \beta \gamma }&:=&\ \mathbf{\tilde{R}}_{\ \beta \gamma
\alpha }^{\alpha }=\mathbf{e}_{\alpha }\tilde{\mathbf{\Gamma }}_{\ \beta
\gamma }^{\alpha }-\mathbf{e}_{\gamma }\ \tilde{\mathbf{\Gamma }}_{\ \beta
\alpha }^{\alpha } 
+\tilde{\mathbf{\Gamma }}_{\ \beta \gamma }^{\varphi }\ \tilde{\mathbf{%
\Gamma }}_{\ \varphi \alpha }^{\alpha }-\tilde{\mathbf{\Gamma }}_{\ \beta
\alpha }^{\varphi }\ \tilde{\mathbf{\Gamma }}_{\ \varphi \gamma }^{\alpha }+%
\tilde{\mathbf{\Gamma }}_{\ \beta \varphi }^{\alpha }W_{\gamma \alpha
}^{\varphi } \\
\mathbf{\tilde{Z}}_{\ \beta \gamma } &:=&\mathbf{\tilde{Z}}_{\ \beta \gamma
\alpha }^{\alpha }=\mathbf{e}_{\alpha }\ \mathbf{\tilde{Z}}_{\ \beta \gamma
}^{\alpha }-\mathbf{e}_{\gamma }\ \mathbf{\tilde{Z}}_{\ \beta \alpha
}^{\alpha }+\mathbf{\tilde{Z}}_{\ \beta \gamma }^{\varphi }\ \mathbf{\tilde{Z%
}}_{\ \varphi \alpha }^{\alpha }-\mathbf{\tilde{Z}}_{\ \beta \alpha
}^{\varphi }\ \mathbf{\tilde{Z}}_{\ \varphi \gamma }^{\alpha }+ \\
&&\tilde{\mathbf{\Gamma }}_{\ \beta \gamma }^{\varphi }\ \mathbf{\tilde{Z}}%
_{\ \varphi \alpha }^{\alpha }-\tilde{\mathbf{\Gamma }}_{\ \beta \alpha
}^{\varphi }\ \mathbf{\tilde{Z}}_{\ \varphi \gamma }^{\alpha }+\mathbf{%
\tilde{Z}}_{\ \beta \gamma }^{\varphi }\ \tilde{\mathbf{\Gamma }}_{\ \varphi
\alpha }^{\alpha }-\mathbf{\tilde{Z}}_{\ \beta \alpha }^{\varphi }\ \tilde{%
\mathbf{\Gamma }}_{\ \varphi \gamma }^{\alpha }+\mathbf{\tilde{Z}}_{\ \beta
\varphi }^{\alpha }W_{\gamma \alpha }^{\varphi }.
\end{eqnarray*}%
}
Introducing in above formulas $\ \mathbf{\tilde{Z}}=\ ^{B}\mathbf{\tilde{Z},}
$ or $\ \mathbf{\tilde{Z}}=\ ^{Ch}\mathbf{\tilde{Z},}$ we get explicit
formulas for distortions of the Ricci tensor (\ref{districci}) for the Berwald,
or Chern, d--connection (for simplicity, we omit such technical results in
this work). Equivalent distortions can be computed if we fix, for instance,
as a ''background'' connection just the Levi--Civita connection $\nabla $ but such constructions are not
adapted to the N--connection splitting. Other fundamental geometric objects
derived for $\ ^{B}\mathbf{D}$ and/or $\ ^{Ch}\mathbf{D}$ can be generated
by noholonomic deformations from analogous ones for the metric compatible,
and almost K\"{a}hler, d--connection $\ ^{\theta }\mathbf{D\equiv \tilde{D}.}
$ This property is very important because it allows us to construct, for
instance, Dirac operators and define generalized Perelman's functionals (see
next section) even such geometric models are metric noncompatible.

We conclude that metric compatible Finsler geometry models with
d--connections uniquely defined by respective metric structures (and induced
by fundamental Finsler functions and Hessians) play a preferred role both
for elaborating geometric and physical theories on $TM.$ Using one of the
connections $\mathbf{\tilde{D},}$ or $\mathbf{\hat{D}}$, we work as in usual
Riemann--Cartan geometry and/or the Ricci flow theory of Riemannian metrics.
It is also possible to reformulate the theories as almost K\"{a}hler
geometries. Then, such constructions can be nonholonomically deformed into
metric noncompatible structures by considering respective distortion tensors.

\subsubsection{A Ricci tensor constructed by Akbar--Zadeh}

For a class of geometric and flow models on $T_{x}M,$ see Refs. \cite%
{akbar,bao1,bao2,peyghan}, a different than (\ref{districci}) Ricci tensor
is used. As noted above, there is an alternative curvature tensor $\breve{R}%
_{k}^{i}$ (\ref{flcurv}) completely determined by semi--spray $\tilde{G}^{k}$
(\ref{smspr}). Contracting indices, we introduce a scalar function $\breve{R}%
(x,y):=F^{-2}\breve{R}_{i}^{i}$ and define a variant of the Ricci tensor,
\begin{equation}
\breve{R}ic_{jk}:=F^{-2}\frac{\partial ^{2}\breve{R}}{\partial y^{j}\partial
y^{k}}.  \label{ricciaz}
\end{equation}

The geometric object $\breve{R}ic_{jk}$ (\ref{ricciaz}) is induced by the
Finsler metric $F$ \ via inverse Hessian ${\tilde{g}}^{ij},$ see ${\tilde{g}}%
_{ij}$ (\ref{hessian}), and $\tilde{G}^{k}$ not involving in such a model
the N--connection structure, lifts on metrics on total space of $TM,$ and
d--connections. By definition, the scalar $\breve{R}$ is positive homogeneous
of degree $0$ in $v$--variables $y^{a}.$

Following such an approach to Finsler geometry, the Einstein metrics ${%
\tilde{g}}_{ij}$ are those for which $\breve{R}ic_{jk}=\lambda (x){\tilde{g}}%
_{jk}$, i.e. when the scalar function $\breve{R}(x,y)=\lambda (x)$ is a
function only on $h$--variables $x^{k}.$ This class of Finsler spaces is by
definition different from that derived for a Ricci d--tensor $\ ^{F}\mathcal{%
R}\mathit{ic}$ (\ref{districci}) on $TM.$ The priority of $\breve{R}ic_{jk}$
(\ref{ricciaz}) is that it is always symmetric (by definition) and
''simplified'' to consider a Ricci field and/or evolution dynamics in any
point $T_{x}M.$ Nevertheless, such a nonholonomically constrained model does
not allow us to study, for instance, mutual transforms of Riemann and
Finsler metrics with general nonsymmetric Ricci d--tensor $\ \mathbf{\tilde{R%
}}_{\ \beta \gamma },$ and respective Einstein d--tensor (\ref{einstfdt}) on
total space of $TM$ and nonholonomic (pseudo) Riemannian manifolds \cite%
{vric3,vric4,vric6} (a series of works from 2006--2008).

A variant of geometric evolution equations for Finsler metrics $F(\chi ,x,y)$
using $\breve{R}ic_{jk}$ was published in 2007 in Ref. \cite{bao1},%
\begin{equation}
\frac{\partial {\tilde{g}}_{ij}}{\partial \chi }=-2\breve{R}ic_{jk},
\label{riccib}
\end{equation}%
when the Hessian ${\tilde{g}}_{ij}(\chi ,x,y)$ and the volume element
\begin{equation}
\upsilon :=(\partial F/\partial y^{i})dx^{i}  \label{volel}
\end{equation}%
depend on $\chi \in \lbrack -\epsilon ,\epsilon ]\subset \mathbb{R}$ and $%
\epsilon >0$ is sufficiently small. These  equations consist an example of
heuristic Finsler evolution equations (\ref{riccifinslerheur}) when $\
^{F}Ric_{ij}\sim \breve{R}ic_{jk}.$ In order to elaborate a self--consistent
Ricci flow theory, at first steps, we have to prove the conditions when $%
\breve{R}ic_{jk}\sim \ ^{F}\Delta ,$ for a Finsler Laplacian, and find
certain analogs of Perelman's functionals from which (\ref{riccib}). \ One
of the aims of the present paper is to show that this type of evolution
models belong to a class of nonholonomicaly constrained systems (in general,
with metric noncompatible Finsler connections) which can be uniquely defined
via corresponding nonholonomic deformations and constraints from theories
flows of the canonical and/or Cartan d--connection \cite%
{vric1,vric2,vric3,vric4,vric5,vric6,vric7,vric8,vric9,vric10}.

For any $F(\chi )=F(\chi ,x,y),$ and respective ${\tilde{g}}_{ij}$ and $%
\tilde{N}_{i}^{a},$ we can compute a family of Ricci d--tensors,
$\mathbf{\tilde{R}}_{\ \beta \gamma }(\chi )=\{\tilde{R}_{\ hj}:=\tilde{R}_{\
hji}^{i},\tilde{P}_{\ ja}:=\tilde{P}_{\ jia}^{i},\ \tilde{S}_{\ bc}:=\
\tilde{S}_{\ bca}^{a}\},$ for the  d--connection $\mathbf{\tilde{D},}$ by constructing
respective tensors in (\ref{curvcart}), or
\begin{equation}
\widehat{\mathbf{R}}_{\ \beta \gamma }(\chi )=\widehat{R}_{ij}\doteqdot \widehat{R}%
_{\ ijk}^{k},\ \ \widehat{R}_{ia}\doteqdot -\widehat{R}_{\ ika}^{k},\
\widehat{R}_{ai}\doteqdot \widehat{R}_{\ aib}^{b},\ \widehat{R}%
_{ab}\doteqdot \widehat{R}_{\ abc}^{c},  \label{riccicdc}
\end{equation}%
for the canonical d--connection $\widehat{\mathbf{D}}$ (\ref{candcon}) (see
explicit formulas for h--v--compo\-nents in Refs. \cite{vrev1,vsgg,ma}).
Because all values $\breve{R}ic_{jk},\tilde{R}_{\ hj}$ and $\widehat{R}_{ij}$
are generated by the same Finsler metric, we can compute in unique forms (up
to frame transforms) the distortions
\begin{equation}
\breve{R}ic_{jk}=\tilde{R}_{\ jk}+\tilde{Z}ic_{jk},\ \breve{R}ic_{jk}=%
\widehat{R}_{\ jk}+\widehat{Z}ic_{jk},  \label{nhdt11}
\end{equation}%
if values $\mathbf{\tilde{R}}_{\ \beta \gamma }$ and $\widehat{\mathbf{R}}%
_{\ \beta \gamma }$ are defined on $TM.$ Similar splitting can be computed
in unique forms for Ricci d--tensors corresponding to $\ \ ^{B}\mathbf{D}%
\mathbf{=}\ \mathbf{\tilde{D}+}\ ^{B}\mathbf{\tilde{Z}}\ $ and $\ ^{Ch}%
\mathbf{D}\mathbf{=}\ \mathbf{\tilde{D}+}\ ^{Ch}\mathbf{\tilde{Z}}$ \ from (%
\ref{distdb}),
\begin{eqnarray}
\ ^{B}\mathbf{R}_{\ \beta \gamma } &=&\mathbf{\tilde{R}}_{\ \beta \gamma }+\
^{B}\mathbf{\tilde{Z}}ic_{\beta \gamma },\ ^{Ch}\mathbf{R}_{\ \beta \gamma }=%
\mathbf{\tilde{R}}_{\ \beta \gamma }+\ ^{Ch}\mathbf{\tilde{Z}}ic_{\beta
\gamma };  \label{nhdt21} \\
\ ^{B}\mathbf{R}_{\ \beta \gamma } &=&\widehat{\mathbf{R}}_{\ \beta \gamma
}+\ ^{B}\widehat{\mathbf{Z}}ic_{\beta \gamma },\ ^{Ch}\mathbf{R}_{\ \beta
\gamma }=\widehat{\mathbf{R}}_{\ \beta \gamma }+\ ^{Ch}\widehat{\mathbf{Z}}%
ic_{\beta \gamma }.  \label{nhdt22}
\end{eqnarray}

If we construct a geometric evolution model for $F(\chi )$ with ${\tilde{g}}%
_{ij}(\chi )$ derived for $\mathbf{\tilde{R}}_{\ \beta \gamma }$, such
constructions are preferred for almost K\"{a}hler models and deformation
quantization \cite{vdefqfl}, or for $\widehat{\mathbf{R}}_{\ \beta \gamma }$
(this is important to study evolution of exact solutions in gravity
theories \cite{vex3,vexsol}), we can always ''extract'' and follow evolution
of geometric objects and metric noncompatible Finsler geometries and/or with
''nonstandard'' curvature (\ref{flcurv}).

\subsection{Einstein--Finsler spaces}

We can work equivalently with any d--connection $\mathbf{\tilde{D},}$ $\mathbf{\hat{D},}%
\ ^{B}\mathbf{D,}$ $\ ^{Ch}\mathbf{D}$ (all these geometric objects are
uniquely determined by $F$ and/or $\widetilde{\mathbf{g}}.$ To study
possible physical applications with generalized gravitational field/
evolution equations is important to decide which type of connection and
nonholonomic constraints are used for elaborating physical theories.

It should be noted that all constructions provided in previous sections can
be performed not only on tangent bundle $TM$ with N--connection splitting (%
\ref{whitneyt}) but on any manifold $\mathbf{V}$ with ''conventional'' $h$--$%
v$--splitting (called also as a nonholonomic manifold) defined by a Whitney
sum
\begin{equation}
T\mathbf{V}=h\mathbf{V}\oplus v\mathbf{V.}  \label{whitney}
\end{equation}%
Such a nonintegrable distribution, for instance, can be introduced always on
a Lorenz manifold $\mathbf{V}$ in general relativity (GR) defining a
so--called $2+2$ splitting.\footnote{%
we use ''boldface'' letters for manifolds, bundles endowed with
N--connection structure and for geometric objects adapted to corresponding
h-v--splitting} More than that, GR and various modifications can be
described equivalently in Finsler and/or almost K\"{a}hler variables, see
details in Refs. \cite{vrev1,vaxiom,vsgg,vex3,vexsol}. There is an unified
formalism for geometrical/physical models which can be elaborated any
nonholonomic manifold, $\mathbf{V,}$ or tangent bundle space, $\mathbf{V}%
=TM. $ Physically, the $y$--variables are treated differently: On a general $%
\mathbf{V},$ such values/coordinates are certain nonholonomically
constrained ones; on $TM,$ the values $y^{a}$ as some ''velocities''
(for dual configurations on $T^{\ast }M,$ there are considered ''momenta'').

The Einstein equations in GR were postulated in standard form using the
Levi--Civita connection  $\nabla =\{\Gamma _{\ \alpha \beta
}^{\gamma }\},$
\begin{equation}
R_{\ \beta \delta }-\frac{1}{2}g_{\beta \delta }R=\varkappa T_{\beta \delta
}.  \label{einsteq}
\end{equation}%
 In formulas (\ref{einsteq}), $R_{\ \beta \delta }$ and $R$
are respectively the Ricci tensor and scalar curvature of $\nabla ;$ it is
also considered the energy--momentum tensor for matter, $T_{\alpha \beta },$
where $\varkappa =const.$ Various tetradic, spinor, connection etc variables
were used with various purposes to construct exact solutions and quantize
gravity, see standard monographs \cite{misner,wald}.

Using conventional
Finsler variables, the gravitational field equations (\ref{einsteq}) can be
re--written equivalently using the canonical d--connection $\widehat{\mathbf{%
D}}$ (\ref{candcon}),
\begin{eqnarray}
\widehat{\mathbf{R}}_{\ \beta \delta }-\frac{1}{2}\mathbf{g}_{\beta \delta
}\ \ _{s}\widehat{R} &=&\widehat{\mathbf{\Upsilon }}_{\beta \delta },
\label{cdeinst} \\
\widehat{L}_{aj}^{c}=e_{a}(N_{j}^{c}),\ \widehat{C}_{jb}^{i}=0,\ \Omega _{\
ji}^{a} &=&0,  \label{lcconstr}
\end{eqnarray}%
for \ $\widehat{\mathbf{\Upsilon }}_{\beta \delta }\rightarrow T_{\beta
\delta }$ if $\widehat{\mathbf{D}}\rightarrow \nabla .$ The constraints (\ref%
{lcconstr}) are equivalent to the condition of vanishing of torsion (\ref%
{dtors}), the distortion d--tensors $\widehat{\mathbf{Z}}=0,$ which
results in $\widehat{\mathbf{D}}\mathbf{=}\nabla ,$ see formulas
(\ref{dcdc}). The system of equations (\ref{cdeinst}) and (\ref{lcconstr}) have a very
important property of decoupling with respect to N--adapted frames (\ref%
{dder}) and (\ref{ddif}) which allows to integrate the Einstein and
geometric evolution equations in very general forms \cite%
{vrev1,vaxiom,vsgg,vex3,vexsol,vric1,vric2,vric8,vric9,vric10}.\footnote{%
Up to frame/coordinate transforms the equations (\ref{einsteq}), and/or (\ref%
{cdeinst}) and (\ref{lcconstr}), are equivalent to
\begin{eqnarray}
\mathbf{\tilde{R}}_{\ \beta \delta }-\frac{1}{2}\widetilde{\mathbf{g}}%
_{\beta \delta }\ _{s}\tilde{R} &=&\mathbf{\tilde{\Upsilon}}_{\beta \delta },
\label{carteinst} \\
\tilde{L}_{\ bi}^{a}=e_{b}(\tilde{N}_{i}^{a}),\ \tilde{C}_{\ jc}^{i}=0,\
\tilde{\Omega}_{ij}^{a} &=&0,  \label{lccnconstcart}
\end{eqnarray}%
when the d--connection is chosen to be the Cartan one $\ ^{\theta }\mathbf{%
D\equiv \tilde{D}.}$ The conditions (\ref{lccnconstcart}) are for zero
torsion (\ref{torscdc}) when $\mathbf{\tilde{Z}}$ $=0$ and $\mathbf{\tilde{D}%
=}\tilde{\nabla}$, in (\ref{dcartdc}). Here, we note that, in general, $%
\mathbf{\tilde{\Upsilon}}_{\beta \delta }$ is different from $\widehat{%
\mathbf{\Upsilon }}_{\beta \delta }.$ The priority of system (\ref{carteinst}%
) written in Cartan d--metric and d--connection Finsler variables is a the
possibility to re--define the geometric objects in almost K\"{a}hler
variables with a further deformation quantization \cite{vdefqfl}. The
''decoupling effect'' for gravitational field equations also exists but the
zero torsion conditions seem to be ''more rigid'' for such configurations.}

On $TM,$ for metric compatible Finsler geometry models, constraints of type (%
\ref{lcconstr}), or (\ref{lccnconstcart}), are not necessary. Using
distortion relations (\ref{dcdc}), (\ref{dcartdc}) and (\ref{distdb}), we
can compute other types of distortions,
\begin{equation}
\tilde{\nabla}=\ ^{Ch}\mathbf{D-}\ _{\nabla }^{Ch}\mathbf{Z=}\ ^{B}\mathbf{D-%
}\ _{\nabla }^{B}\mathbf{Z=\tilde{D}-\tilde{Z}=}\widehat{\mathbf{D}}-%
\widehat{\mathbf{Z}},  \label{distrelg}
\end{equation}%
where all geometric objects are determined by $F(u)$ via ${\tilde{g}}%
_{ij}(u).$ Such nonholonomic constraints show that in any model of Finsler geometry we can consider equivalently ''not--adapted'' (to N--connection) geometric constructions with $\tilde{\nabla}$ defined by a (pseudo)
Riemannian metric
\begin{equation}
\ \tilde{g}_{\underline{\alpha }\underline{\beta }}=\left[
\begin{array}{cc}
\ \tilde{g}_{ij}+\tilde{N}_{i}^{a}~\tilde{N}_{j}^{b}\ \tilde{g}_{ab} & ~%
\tilde{N}_{j}^{e}\ \tilde{g}_{ae} \\
~\tilde{N}_{i}^{e}\ \tilde{g}_{be} & \tilde{g}_{ab}%
\end{array}%
\right] ,  \label{offd}
\end{equation}%
where the coefficients $\tilde{g}_{\underline{\alpha }\underline{\beta }}$
are those for the Finsler d--metric (\ref{slm}) re--defined with respect to
a coordinate co-basis, $du^{\underline{\alpha }}=(dx^{\underline{i}},dy^{%
\underline{a}}).$ The nonholonomic structure is encoded into vielbeins $%
\mathbf{\tilde{e}}_{\alpha }=\ \mathbf{\tilde{e}}_{\alpha }^{\ \underline{%
\alpha }}(u)\partial _{\underline{\alpha }}$ with coefficients
\begin{equation}
\ \mathbf{\tilde{e}}_{\alpha }^{\ \underline{\alpha }}(u)=\left[
\begin{array}{cc}
\ \tilde{e}_{i}^{\ \underline{i}}(u) & ~\tilde{N}_{i}^{b}(u)\ \tilde{e}%
_{b}^{\ \underline{a}}(u) \\
0 & \tilde{e}_{a}^{\ \underline{a}}(u)%
\end{array}%
\right] ,\   \label{vielb}
\end{equation}%
when $\tilde{g}_{ij}(u)=\tilde{e}_{i}^{\ \underline{i}}(u)\ \tilde{e}_{j}^{\
\underline{j}}(u)\eta _{\underline{i}\underline{j}}$, for $\eta _{\underline{%
i}\underline{j}}=diag[\pm 1,...\pm 1]$ fixing a corresponding local metric
signature on $TM.$

We conclude this section with the remark that the models of Finsler geometry on $%
TM$ \footnote{and/or any nonholonomic manifold $\mathbf{V}$ with N--connection splitting} with Cartan/ canonical d--connection, Berwald and/or Chern d--connections can be reconsidered equivalently as certain nonholonomic (pseudo) Riemannian ones endowed with nonholonomic $h$--$v$--splitting and corresponding unique distortions of $\tilde{\nabla}$. The distortion relations (\ref{distrelg})
play a crucial role in constructing models of Finsler--Ricci flow evolution uniquely related to standard theory of Ricci flows for Riemannian
geometries. The main theorems can be proven using $\tilde{\nabla}$ and then the results for Finsler flows are stated by ''uniquely'' defined
nonholonomic distortions and constraints.

\section{Finsler--Ricci Flows and  Distortions}
\label{s4}
In this section we show how a self--consistent approach to geometric flows with metric noncompatible connections can be elaborated if there are used special classes of  nonholonomic deformations/distortions of metric compatible flows.

\subsection{The Perelman's Functionals on Finsler Spaces}

 G. Perelman's idea \cite{gper1} was to derive the Ricci flow
equations of (pseudo) Riemannian geometries as gradient flows for some
functionals defined by the Levi--Civita connection $\nabla $ and respective scalar curvature $\ _{\nabla}R$. Considering a compact region $\mathcal{V}\subset TM$ (in general, we can take any
nonholonomic manifold $\mathbf{V}$ instead of $TM$),
with $\tilde{\nabla}$  computed for
$\ \tilde{g}_{\underline{\alpha }\underline{\beta }}$
(\ref{offd}). This family of geometric objects is induced by a family of Finsler generating function $F(\tau ,x,y)$ parametrized
by a flow parameter
$\tau \in \lbrack -\epsilon ,\epsilon ]\subset \mathbb{R}$ with a
sufficiently small $\epsilon >0.$ It is possible to introduce such
functionals in Finsler geometry (we use our system of denotations),
\begin{eqnarray}
\ _{\shortmid }\mathcal{F}(\mathbf{\tilde{g}},\tilde{\nabla},f) &=&\int_{%
\mathcal{V}}\left( \ _{\nabla }\tilde{R}+\left| \tilde{\nabla}f\right|
^{2}\right) e^{-f}\ dV,  \label{2pfrs} \\
\ _{\shortmid }\mathcal{W}(\mathbf{\tilde{g}},\tilde{\nabla},f,\tau )
&=&\int_{\mathcal{V}}\left[ \tau \left( \ _{\nabla }\tilde{R}+\left| \tilde{%
\nabla}f\right| \right) ^{2}+f-2n)\right] \mu \ dV,  \nonumber
\end{eqnarray}%
where $dV$ is the volume form of $\ \mathbf{g\sim \tilde{g}}$
(up to frame transforms), integration is taken over
$\mathcal{V},\dim \mathcal{V}=2n.$
Via frame transforms and for a parameter $\tau >0,$ we can fix $\int_{%
\mathcal{V}}dV=1$ when $\mu =\left( 4\pi \tau \right) ^{-n}e^{-f}.$  Working with $\tilde{\nabla},$ we can model in "not N--adapted" form different types of Ricci flow evolutions of Finsler geometries by imposing nonholonomic constraints with a distortion relation (\ref{distrelg}). In this approach, the Finsler--Ricci flows can be considered as evolving nonholonomic dynamical systems on the space of Riemannian metrics on $TM$ and the functionals $\ _{\shortmid }\mathcal{F}$ and $\ _{\shortmid }\mathcal{W}$ are of Lyapunov type. Levi--Civita Ricci flat configurations are defined as ''fixed'' on $\tau $ points of the corresponding dynamical systems.

The goal of this section is to re--define the functionals (\ref{2pfrs}) in
N--adapted form when the evolution of Finsler geometries with Sasaki
type metrics (\ref{slm}) on $\widetilde{TM}$ will be extracted
by a corresponding fixing $\ ^{F}\mathbf{D=}\ ^{Ch}\mathbf{D,}$
or$\ ^{B}\mathbf{D}$ (the
variants with $\ ^{F}\mathbf{D=}\ \mathbf{\tilde{D},}$ and/or $\mathbf{=}%
\widehat{\mathbf{D}}\mathbf{\ }$ where studied in Refs. \cite{vric4,vric6}).

\begin{lemma}
For a Finsler geometry model with d--connection $\ ^{F}\mathbf{D}$
completely determined by $F$ and $\mathbf{\tilde{g},}$ the Perelman's
functionals (\ref{2pfrs}) can be re--written equivalently in N--adapted form
by considering distortion relations for scalar curvature and Ricci tensor (%
\ref{districci}), {\small
\begin{eqnarray}
\ ^{F}\mathcal{F}(\mathbf{\tilde{g}},\ ^{F}\mathbf{D},\breve{f}) &=&\int_{%
\mathcal{V}}(\ _{s}^{F}R+|\ ^{F}\mathbf{D}\breve{f}|^{2})e^{-\breve{f}}\ dV,
\label{2npf1} \\
\ ^{F}\mathcal{W}(\mathbf{\tilde{g}},\ ^{F}\mathbf{D},\breve{f},\breve{\tau}%
) &=&\int_{\mathcal{V}}[\breve{\tau}(\ _{s}^{F}R+|\ ^{h}D\breve{f}|+|\ ^{v}D%
\breve{f}|)^{2}+\breve{f}-2n]\breve{\mu}dV,  \label{2npf2}
\end{eqnarray}%
} where the scalar curvature $\ _{s}^{F}R$ (\ref{riccifs}) is computed for $%
\ ^{F}\mathbf{D}=(\ _{h}^{F}D,\ \ _{v}^{F}D)$, $\left| \ ^{F}\mathbf{D}%
\breve{f}\right| ^{2}=\left| \ _{h}^{F}D\breve{f}\right| ^{2}+\left| \
_{v}^{F}D\breve{f}\right| ^{2},$ and the new scaling function $\breve{f}$
satisfies $\int_{\mathcal{V}}\breve{\mu}dV=1$ for $\breve{\mu}=\left( 4\pi
\breve{\tau}\right) ^{-n}e^{-\breve{f}}$ and $\breve{\tau}>0.$
\end{lemma}

\begin{proof}
The proof of this Lemma is similar to that for Claim 3.1 in Ref. \cite{vric4}
for nonholonomic manifolds (for a prescribed canonical d--connection). On $%
\widetilde{TM},$ such a statement transforms into a Lemma similar to that in
original Perelman's work \cite{gper1} if we consider models of Finsler
geometry with $\ ^{F}\mathbf{D}$  related to $\tilde{\nabla}$ via a unique
distortion relation (\ref{distrelg}). For simplicity, we can use $\breve{\tau}%
=\ ^{h}\tau =\ ^{v}\tau $ for a couple of possible $h$-- and $v$--flows
parameters, $\breve{\tau}=(\ ^{h}\tau ,\ ^{v}\tau ),$ and introduce a new
function $\breve{f},$ instead of $f.$ This scalar function is re--defined in
such a form that in formulas (\ref{2pfrs}) the distortion of Ricci tensor (%
\ref{districci}) and d--connection under $\tilde{\nabla}\rightarrow \ ^{F}%
\mathbf{D}$ results in
\begin{equation}
(\ _{\shortmid }\tilde{R}+|\tilde{\nabla}f|^{2})e^{-f}=(\ _{s}^{F}R+|\ ^{F}%
\mathbf{D}\breve{f}|^{2})e^{-\breve{f}}\ +\Phi  \label{rdef1}
\end{equation}%
for (\ref{2npf1}). Similarly, we re--scale the parameter $\tau \rightarrow
\breve{\tau}$ to have
\begin{equation}
\lbrack \tau (\ _{\shortmid }\tilde{R}+|\tilde{\nabla}f|)^{2}+f-2n)]\mu =[%
\breve{\tau}(\ _{s}^{F}R+|\ ^{h}D\breve{f}|+|\ ^{v}D\breve{f}|)^{2}+\breve{f}%
-2n]\breve{\mu}+\Phi _{1}  \label{rdef2}
\end{equation}%
for some $\Phi $ and $\Phi _{1}$ for which $\int_{\mathcal{V}}\Phi dV=0$
and $\int_{\mathcal{V}}\Phi _{1}dV=0.$ This results in
formula (\ref{2npf2}). Finally, in this proof, we conclude that both
in metric compatible and
noncompatible Finsler models uniquely determined by $F$ and $\mathbf{\tilde{g%
}}$, the Perelman functionals are certain nonholonomic deformations of those for $\tilde{\nabla}.$
\end{proof}

A similar proof with redefinition to a corresponding function $\breve{f}%
\rightarrow \underline{f}$ and parameter $\breve{\tau}\rightarrow \underline{%
\tau },$ can be used for proof of

\begin{corollary}
Fixing a point $x\in TM$ and a compact region $\mathcal{V}_{x}$ and via
distortions (\ref{nhdt11}), respectively, we can transform (\ref{2npf1}) and
(\ref{2npf2}) into {\small
\begin{eqnarray}
\ ^{F}\mathcal{F}(\tilde{g}_{ij},\ \tilde{D},\breve{f}) &=&\int_{\mathcal{V}%
}(\ \tilde{g}^{jk}\breve{R}ic_{jk}+|\tilde{D}\breve{f}|^{2})e^{-\underline{f}%
}\ dV,  \label{npf1} \\
\ ^{F}\mathcal{W}(\tilde{g}_{ij},\ \tilde{D},\breve{f},\breve{\tau})
&=&\int_{\mathcal{V}}[\underline{\tau }(\ \tilde{g}^{jk}\breve{R}ic_{jk}+|%
\tilde{D}\breve{f}|)^{2}+\underline{f}-n]\breve{\mu}dV,  \label{npf2}
\end{eqnarray}%
}defining a nonholonomic dynamics  related to Akbar--Zadeh definition of the
Ricci tensor  $\breve{R}ic_{jk}$ (\ref{ricciaz}).
\end{corollary}

In above formulas, integrals of type $\int_{\mathcal{V}}\{\ldots \}dV$ can
be transformed into computations on ''spherical'' bundle $SM,$ see details %
\cite{akbar,bao2,peyghan},
 $$\int_{SM}\{\ldots \}\frac{(-1)^{n(n-1)/2}}{(n-1)!}\upsilon \wedge (d\upsilon
)^{n-1}=\int_{SM}\{\ldots \}dV_{SM},$$
where the volume element $\upsilon $ is determined by $F$ following formula (%
\ref{volel}).

\subsection{On N--adapted geometric structures}

 Any geometric configuration and Ricci flow evolution formula for Riemannian metrics containing the Levi--Civita connection $\nabla $ can be transformed into its analogous on $TM$ for Finsler spaces following such rules:
\begin{enumerate}
\item Consider a $h$-$v$--splitting determined by $F(\chi ):=F(\chi ,u)$ via
flows of canonical N--connection $\mathbf{\tilde{N}}$ $(\chi )$ and adapted
frames $$
\partial _{\underline{\alpha }} \rightarrow \mathbf{e}_{\alpha }(\chi )=(%
\mathbf{e}_{i}(\chi )=\partial _{i}-N_{i}^{b}(\chi )\partial
_{b},e_{a}=\partial _{a}),$$ $$
du^{\underline{\alpha }} \rightarrow \mathbf{e}^{\alpha }(\chi )=\left(
e^{i}=dx^{i},e^{a}=dy^{a}+N_{k}^{a}(\chi )dx^{k}\right),$$
related to $\mathbf{\tilde{e}}_{\alpha }(\chi )$ (\ref{dder}) and $\mathbf{%
\tilde{e}}^{\alpha }(\chi )$ (\ref{ddif}) by any convenient
frame transforms.

\item Metrics $\ \tilde{g}_{\underline{\alpha }\underline{\beta }}$ $(\chi )$
(\ref{offd}) are transformed equivalently into d--metrics $\mathbf{\tilde{g}}%
(\chi )$ (\ref{slm}) and/or any related via frame transforms $\mathbf{g}$ (%
\ref{m1}).

\item Via distortion relations (\ref{distrelg}), we construct necessary
chains of distortions of connections,
  $\tilde{\nabla}(\chi )\rightarrow \nabla (\chi )\rightarrow \widehat{\mathbf{D%
}}(\chi )=\left( \ ^{h}D(\chi ),\ ^{v}D(\chi )\right) \rightarrow \ ^{F}%
\mathbf{D}(\chi )$,
where $\ ^{F}\mathbf{D=\tilde{D},=}\ ^{Ch}\mathbf{D,}$ or $\mathbf{=}\ ^{B}%
\mathbf{D.}$

\item Using such distortions of connections, we can compute distortions of
curvature tensors and related Ricci \ tensors (see (\ref{districci}), (\ref%
{nhdt21}), (\ref{nhdt22}) and (\ref{nhdt11})) and scalar curvatures.

\item Changing data $(f,\tau )\rightarrow (\breve{f},\breve{\tau})$ given by
formulas of type (\ref{rdef1}) and (\ref{rdef2}), we compute distortions of
the Perelman's functionals (\ref{2pfrs}), i.e $\ _{\shortmid }\mathcal{F}$
and $\ _{\shortmid }\mathcal{W}$ $,$ into $\ ^{F}\mathcal{F}$ \ and $\ ^{F}%
\mathcal{W},$ respectively, (\ref{2npf1}) and (\ref{2npf2}).
\end{enumerate}

In this work, we shall omit detailed proofs if they can be obtained using
metric compatible constructions in (pseudo) Riemannian and Lagrange--Finsler
geometry as in Refs. \cite{gper1,caozhu,vric3,vric4,vric6} following above stated rules.

\subsection{Hamilton equations for  metric noncompatible Finsler spa\-ces}

For the canonical d--connection $\widehat{\mathbf{D}}$ (similarly, for $\
\mathbf{\tilde{D})}$, we can construct the canonical Laplacian operator, $%
\widehat{\Delta }:=$ $\widehat{\mathbf{D}}$ $\widehat{\mathbf{D}}$ , h- and
v--components of the Ricci tensor, $\widehat{R}_{ij}$ and $\widehat{R}_{ab},$
and consider parameter $\tau (\chi ),$ $\partial \tau /\partial \chi =-1$
(for simplicity, we do not include the normalized term).

\begin{theorem}
\label{2theq1}The Finsler--Ricci flows for $\ ^{F}\mathbf{D}$ preserving a symmetric metric structure $\mathbf{g=\tilde{g}}$ and nonholonomic
constraints
\begin{equation}
\tilde{\nabla}=\ ^{F}\mathbf{D-}\ _{\nabla }^{F}\mathbf{Z,\ }\widehat{%
\mathbf{D}}=\ ^{F}\mathbf{D-}\ ^{F}\widehat{\mathbf{Z}},  \label{dista1}
\end{equation}%
resulting in distortions
\begin{eqnarray}
\widehat{\Delta } &=&\widehat{\mathbf{D}}_{\alpha }\ \widehat{\mathbf{D}}%
^{\alpha }=\ ^{F}\Delta +\ ^{Z}\widehat{\Delta },  \label{dista2} \\
\ ^{F}\Delta &=&\ ^{F}\mathbf{D}_{\alpha }\ ^{F}\mathbf{D}^{\alpha },\ \ ^{Z}%
\widehat{\Delta }=\ ^{F}\widehat{\mathbf{Z}}_{\alpha }\ ^{F}\widehat{\mathbf{%
Z}}^{\alpha }-[\ ^{F}\mathbf{D}_{\alpha }(\ ^{F}\widehat{\mathbf{Z}}^{\alpha
})+\ ^{F}\widehat{\mathbf{Z}}_{\alpha }(\ ^{F}\mathbf{D}^{\alpha })];  \nonumber
\\
\widehat{\mathbf{R}}_{\ \beta \gamma } &=&\ ^{F}\mathbf{R}_{\ \beta \gamma
}-\ ^{F}\widehat{\mathbf{Z}}ic_{\beta \gamma },\ \ _{s}\widehat{R}=\
_{s}^{F}R-\mathbf{g}^{\beta \gamma }\ ^{F}\widehat{\mathbf{Z}}ic_{\beta
\gamma }=\ _{s}^{F}R-\ _{s}^{F}\widehat{\mathbf{Z}},  \nonumber \\
\ _{s}^{F}\widehat{\mathbf{Z}}&= &\mathbf{g}^{\beta \gamma }\ ^{F}\widehat{%
\mathbf{Z}}ic_{\beta \gamma }=\ _{h}^{F}\widehat{Z}+\ _{v}^{F}\widehat{Z},\
_{h}^{F}\widehat{Z}=g^{ij}\ ^{F}\widehat{\mathbf{Z}}ic_{ij},\ _{v}^{F}%
\widehat{Z}=h^{ab}\ ^{F}\widehat{\mathbf{Z}}ic_{ab};  \nonumber \\
\ _{s}^{F}R &=&\ _{h}^{F}R+\ _{v}^{F}R,\ \ _{h}^{F}R:=g^{ij}\ ^{F}R_{ij},\
_{v}^{F}R=h^{ab}\ ^{F}R_{ab},  \nonumber
\end{eqnarray}%
can be characterized by two equivalent systems of geometric flow  equations:
\begin{enumerate}
\item Evolution with distortions of the canonical d--connections introduced
in metric compatible nonholonomic Ricci flow equations,
\begin{eqnarray}
\frac{\partial g_{ij}}{\partial \chi } &=&-2\left( \ ^{F}\mathbf{R}_{\ ij}-\
^{F}\widehat{\mathbf{Z}}ic_{ij}\right) ,\ \frac{\partial g_{ab}}{\partial
\chi }=-2\left( \ ^{F}\mathbf{R}_{\ ij}-\ ^{F}\widehat{\mathbf{Z}}%
ic_{ij}\right),  \nonumber \\
\ ^{F}\mathbf{R}_{\ ia} &=&\ ^{F}\widehat{\mathbf{Z}}ic_{ia},\ \ ^{F}\mathbf{%
R}_{\ ai}=\ ^{F}\widehat{\mathbf{Z}}ic_{ai},\   \label{frham1a} \\
\ \frac{\partial \widehat{f}}{\partial \chi } &=&-\left( \ ^{F}\Delta +\ ^{Z}%
\widehat{\Delta }\right) \widehat{f}+\left| \left( \ ^{F}\mathbf{D}-\ ^{F}%
\widehat{\mathbf{Z}}\right) \widehat{f}\right| ^{2}-\ _{s}^{F}R+\ _{s}^{F}%
\widehat{\mathbf{Z}},  \nonumber
\end{eqnarray}%
and
 {\small $\frac{\partial }{\partial \chi }\widehat{\mathcal{F}}(\mathbf{g,\mathbf{\ }%
\widehat{\mathbf{D}},}\widehat{f}) =
 2\int_{\mathcal{V}}[|\ ^{F}R_{\ ij}-\ ^{F}\widehat{\mathbf{Z}}ic_{ij}+(\
^{F}\mathbf{D}_{i}\mathbf{-}\ ^{F}\widehat{\mathbf{Z}}_{i})(\ ^{F}\mathbf{D}%
_{j}\mathbf{-}\ ^{F}\widehat{\mathbf{Z}}_{j})\widehat{f}|^{2}+ |\ ^{F}R_{ab}-\ ^{F}\widehat{\mathbf{Z}}ic_{ab}+(\ ^{F}\mathbf{D}_{a}-\
^{F}\widehat{\mathbf{Z}}_{a})(\ ^{F}\mathbf{D}_{b}\mathbf{-}\ ^{F}\widehat{%
\mathbf{Z}}_{b})\widehat{f}|^{2}]e^{-\widehat{f}}dV,$}
when $\int_{\mathcal{V}}e^{-\widehat{f}}dV$ is constant.

\item Evolution derived from distorted Perelman's functional $\ ^{F}\mathcal{%
F}(\mathbf{\tilde{g}},\ ^{F}\mathbf{D},\breve{f})$ (\ref{2npf1}),
\begin{eqnarray}
\frac{\partial g_{ij}}{\partial \chi } &=&-2\left( \ ^{F}\mathbf{R}_{\ ij}-\
^{F}\widehat{\mathbf{Z}}ic_{ij}\right) ,\ \frac{\partial g_{ab}}{\partial
\chi }=-2\left( \ ^{F}\mathbf{R}_{\ ij}-\ ^{F}\widehat{\mathbf{Z}}%
ic_{ij}\right) ,  \nonumber \\
\ ^{F}\mathbf{R}_{\ ia} &=&\ ^{F}\widehat{\mathbf{Z}}ic_{ia},\ \ ^{F}\mathbf{%
R}_{\ ai}=\ ^{F}\widehat{\mathbf{Z}}ic_{ai},\   \label{frham1b} \\
\ \frac{\partial \breve{f}}{\partial \chi } &=&-\ ^{F}\Delta \breve{f}%
+\left| \ ^{F}\mathbf{D}\breve{f}\right| ^{2}-\ _{s}^{F}R\ ,  \nonumber
\end{eqnarray}%
and the property that
   $\frac{\partial }{\partial \chi }\ ^{F}\mathcal{F}(\mathbf{\tilde{g}},\ ^{F}%
\mathbf{D},\breve{f})=$\newline $2\int_{\mathcal{V}}[|\ ^{F}\mathbf{R}_{\ \beta \gamma
}+\ ^{F}\mathbf{D}_{\beta }\ ^{F}\mathbf{D}_{\gamma }\breve{f}|^{2}]e^{-%
\breve{f}}dV,$ when $\int_{\mathcal{V}}e^{-\breve{f}}dV=const$.
\end{enumerate}
\end{theorem}

\begin{proof}
The distortions (\ref{distrelg}) can be written in an equivalent form (\ref%
{dista1}) which allows us to compute respective splitting for Laplacians
and, following formula (\ref{riccifs}), the decomposition of necessary types
Ricci and scalar curvature operators (\ref{dista2}). This reduces the
constructions to a corresponding system of Ricci flow evolution equations
for $\mathbf{\mathbf{\ }\widehat{\mathbf{D}},}$ see proofs in Refs. \cite%
{vric4,vric6},
\begin{equation}
\frac{\partial g_{ij}}{\partial \chi } =-2\widehat{R}_{ij},\ \frac{%
\partial \underline{g}_{ab}}{\partial \chi }=-2\widehat{R}_{ab},\
 \frac{\partial \widehat{f}}{\partial \chi } =-\widehat{\Delta }\widehat{f%
}+\left| \widehat{\mathbf{D}}\widehat{f}\right| ^{2}-\ ^{h}\widehat{R}-\ ^{v}%
\widehat{R},  \label{rfcandc}
\end{equation}%
derived from the functional $\widehat{\mathcal{F}}(\mathbf{\tilde{g}},%
\widehat{\mathbf{D}},\widehat{f})=\int_{\mathcal{V}}(\ _{s}\widehat{R}+|%
\widehat{\mathbf{D}}\widehat{f}|^{2})e^{-\widehat{f}}\ dV$.

Such metric compatible canonical Finsler--Ricci flow equations are
equivalent (via nonholonomic transforms $\nabla \rightarrow \widehat{\mathbf{%
D}}$) to those proposed for Riemannian spaces by G. Perelman \cite{gper1}
(details of the proof with $\nabla $ are given in Proposition 1.5.3 of \cite%
{caozhu}). We must impose the conditions $\widehat{R}_{ia}=0$ and $\widehat{R%
}_{ai}=0$ if we wont to keep the total metric to be symmetric under Ricci
evolution. If such conditions are not satisfied, we generate nonsymmetric
metrics because  the Ricci tensor may be nonsymmetric for Finsler spaces, see details in \cite{vric5}.
The system of equations (\ref{frham1a}) is just that for the canonical
d--connection (\ref{rfcandc}) but rewritten in terms of (in general, metric
noncompatible) $\ ^{F}\mathbf{D.}$ This means that we can follow a metric
noncompatible evolution derived from a Perlman type functional $\widehat{%
\mathcal{F}}$ formulated in terms of the canonical d--connection and
respective scalar function $\widehat{f}.$

Finally, we note that the functional $\widehat{\mathcal{F}}(\mathbf{g,%
\widehat{\mathbf{D}},}\widehat{f})$ is nondecreasing in time and the
monotonicity is strict unless we are on a steady N--adapted gradient
solution (see details in \cite{vric4}). This property may not ''survive''
under nonholonomic deformations to certain $\ ^{F}\mathbf{D}$. This is not
surprising for metric noncompatible geometric evolutions. Such
distortions can be computed in unique forms due to relations (\ref{distrelg}%
) and kept under control via nonholonomic constraints which allows us to
construct  $\ ^{F}\mathcal{F}(\mathbf{\tilde{g}},\ ^{F}\mathbf{D},\breve{f}%
)$ and derive metric noncompatible evolution equations (\ref{frham1b}).
\end{proof}

The above theorem can be reformulated in terms of distortions from the
Cartan d--connection, when $%
\mathbf{\tilde{D}}=\ ^{F}\mathbf{D-}\ ^{F}\mathbf{\tilde{Z}}$ is used in (\ref{dista1})
( instead of $\mathbf{\ }\widehat{\mathbf{D}}=\ ^{F}%
\mathbf{D-}\ ^{F}\widehat{\mathbf{Z}}$).   To consider an almost K\"{a}hler model of Cartan--Finsler space is important because  following such an approach we work with almost symplectic variables, see an explicit construction in section   \ref{akmfg}. This way, it is possible to perform deformation quantization of the Finsler--Ricci flow theory \cite{vdefqfl} and develop noncommutative models \cite{vric7} applying standard geometric quantization methods.

The Finsler--Ricci evolution equations derived in this work are with respect
to N--adapted frames (\ref{dder}) and (\ref{ddif}) which in their turn are
subjected to geometric evolution. Using vielbein parametrizations (\ref%
{vielb}) and similar formulas for Riemannian spaces \cite%
{caozhu,kleiner,rbook} (see also models of geometric evolution with
N--connections in \cite{vric3,vric4}),

\begin{corollary}
The evolution, for all time $\tau \in \lbrack 0,\tau _{0}),$ of N--adapted
frames in a Finsler space,
 $ \mathbf{\tilde{e}}_{\alpha }(\tau )=\ \mathbf{\tilde{e}}_{\alpha }^{\
\underline{\alpha }}(\tau ,u)\partial _{\underline{\alpha }},$
 up to frame/coordinate transforms, is defined by the coefficients
 \begin{eqnarray*} \mathbf{\tilde{e}}_{\alpha }^{\ \underline{\alpha }}(\tau ,u) &=&\left[
\begin{array}{cc}
\ e_{i}^{\ \underline{i}}(\tau ,u) & ~\tilde{N}_{i}^{b}(\tau ,u)\ e_{b}^{\
\underline{a}}(\tau ,u) \\
0 & \ e_{a}^{\ \underline{a}}(\tau ,u)%
\end{array}%
\right] ,\\
\mathbf{\tilde{e}}_{\ \underline{\alpha }}^{\alpha }(\tau ,u)\ &=&\left[
\begin{array}{cc}
e_{\ \underline{i}}^{i}=\delta _{\underline{i}}^{i} & e_{\ \underline{i}%
}^{b}=-\tilde{N}_{k}^{b}(\tau ,u)\ \ \delta _{\underline{i}}^{k} \\
e_{\ \underline{a}}^{i}=0 & e_{\ \underline{a}}^{a}=\delta _{\underline{a}%
}^{a}%
\end{array}%
\right],\end{eqnarray*}
with
 $ \tilde{g}_{ij}(\tau )=\ e_{i}^{\ \underline{i}}(\tau ,u)\ e_{j}^{\
\underline{j}}(\tau ,u)\eta _{\underline{i}\underline{j}}$  and
 $\tilde{g}_{ab}(\tau )=\ e_{a}^{\ \underline{a}}(\tau ,u)\ e_{b}^{\
\underline{b}}(\tau ,u)\eta _{\underline{a}\underline{b}}$,
 where $\eta _{\underline{i}\underline{j}}=diag[\pm 1,...\pm 1]$ and $\eta _{%
\underline{a}\underline{b}}=diag[\pm 1,...\pm 1]$ fix a signature of $\
\mathbf{\tilde{g}}_{\alpha \beta }^{[0]}(u),$ is given by equations
\begin{equation}
\frac{\partial }{\partial \tau }\mathbf{\tilde{e}}_{\ \underline{\alpha }%
}^{\alpha }\ =\ \mathbf{\tilde{g}}^{\alpha \beta }~\widehat{\mathbf{R}}%
_{\beta \gamma }~\ \mathbf{\tilde{e}}_{\ \underline{\alpha }}^{\gamma }
\label{aeq5}
\end{equation}%
if we prescribe that the geometric constructions are derived by the
canonical d--connection.
\end{corollary}

Finally, we emphasize that $\mathbf{g}^{\alpha \beta }~\widehat{\mathbf{R}}%
_{\beta \gamma }=g^{ij}\widehat{R}_{ij}+g^{ab}\widehat{R}_{ab}$ in (\ref%
{aeq5}) selects for evolution only the symmetric components of the Ricci
d--tensor for the canonical d--connection. The formulas for a distortion $%
\widehat{\mathbf{R}}_{\ \beta \gamma }=\ ^{F}\mathbf{R}_{\ \beta \gamma }-\
^{F}\widehat{\mathbf{Z}}ic_{\beta \gamma }$ allow us to compute flow
contributions defined by metric noncompatibe flows with $\ ^{F}\mathbf{R}_{\
\beta \gamma }.$

\subsection{Statistical analogy and thermodynamics of  Finsler--Ricci flows}
The functional $\ _{\shortmid }\mathcal{W}$ is in a\ sense
analogous to minus entropy \cite{gper1} and this property was proven for
metric compatible Finsler--Ricci flows \cite{vric4,vric6} with functionals $%
\widehat{\mathcal{W}}$ \ and/or $\mathcal{\tilde{W}}$, respectively written
for $\widehat{\mathbf{D}}$ and $\mathbf{\tilde{D}}.$ This allows us to
associate some thermodynamical values characterizing (non) holonomic  geometric evolution.
The aim of this section is to show how a statistical/thermodynamic analogy
can be provided for metric noncompatible Ricci flows.

For the functionals $\widehat{\mathcal{W}}$ and $\ ^{F}\mathcal{W}$ (\ref%
{2npf2}), we can prove two systems of equations as in Theorem \ref{2theq1}
(we omit such considerations in this work). For simplicity, we provide an
equivalent result  stated for $\mathcal{\tilde{W}}.$
\begin{theorem}
\label{2theveq}For any d--metric $\mathbf{g}(\chi )$ (\ref{m1}), $\mathbf{%
\tilde{D}}=\ ^{F}\mathbf{D-}\ ^{F}\mathbf{\tilde{Z},}$ and functions $%
\widehat{f}(\chi )$ and $\widehat{\tau }(\chi )$ being  solutions of the
system of equations%
\begin{eqnarray*}\frac{\partial g_{ij}}{\partial \chi } &=& -2\left( \ ^{F}R_{\ ij}-\ ^{F}%
\mathbf{\tilde{Z}}ic_{ij}\right),\
\frac{\partial g_{ab}}{\partial \chi } =-2\left( \ ^{F}R_{\ ab}-\ ^{F}%
\mathbf{\tilde{Z}}ic_{ab}\right),\\
\ \frac{\partial \tilde{f}}{\partial \chi } &=& -(\ ^{F}\Delta +\ ^{Z}\tilde{%
\Delta})\tilde{f}+\left| (\ ^{F}\mathbf{D-}\ ^{F}\mathbf{\tilde{Z}})\tilde{f}%
\right| ^{2}-\ \ _{s}\tilde{R}\ +\frac{2n}{\hat{\tau}},\
\frac{\partial \tilde{\tau}}{\partial \chi } = -1,
\end{eqnarray*}
 it is satisfied the condition
$\frac{\partial }{\partial \chi }\mathcal{\tilde{W}}(\mathbf{g}(\chi )%
\mathbf{,}\tilde{f}(\chi ),\tilde{\tau}(\chi ))=2\int_{\mathcal{V}}\tilde{%
\tau}[|\ ^{F}\mathbf{R}_{\alpha \beta }-\ ^{F}\mathbf{\tilde{Z}}ic_{\alpha
\beta }+ (\ ^{F}\mathbf{D}_{\alpha }\mathbf{-}\ ^{F}\mathbf{\tilde{Z}}_{\alpha })(\
^{F}\mathbf{D}_{\alpha }\mathbf{-}\ ^{F}\mathbf{\tilde{Z}}_{\alpha })\tilde{f%
}-\frac{1}{2\tilde{\tau}}\mathbf{\tilde{g}}_{\alpha \beta }|^{2}](4\pi
\tilde{\tau})^{-n}e^{-\tilde{f}}dV,$
 for $\int_{\mathcal{V}}e^{-\tilde{f}}dV=const$. This functional is
N--adapted nondecreasing if it is both h-- and v--nondecreasing.
\end{theorem}

\begin{proof}
We apply  a proof
with N--adapted modification of Proposition 1.5.8 in \cite{caozhu}
containing the details of the original result from \cite{gper1}). For metric
compatible Lagrange and/or Finsler flows, there are proofs  \cite%
{vric4,vric6} that  for $\widehat{\mathbf{D}},$ the equations
$$
\frac{\partial g_{ij}}{\partial \chi } =-2\widehat{R}_{ij},\ \frac{%
\partial g_{ab}}{\partial \chi }=-2\widehat{R}_{ab},\
\ \frac{\partial \widehat{f}}{\partial \chi } = -\widehat{\Delta }\widehat{f%
}+\left| \widehat{\mathbf{D}}\widehat{f}\right| ^{2}-\ _{s}\widehat{R}+\frac{%
2n}{\hat{\tau}}, \  \frac{\partial \hat{\tau}}{\partial \chi } = -1$$
 result in the condition
 $\frac{\partial }{\partial \chi }\widehat{\mathcal{W}}(\mathbf{g}(\chi )%
\mathbf{,}\widehat{f}(\chi ),\hat{\tau}(\chi ))=2\int_{\mathcal{V}}\hat{\tau}%
[|\widehat{R}_{ij}+\widehat{D}_{i}\widehat{D}_{j}\widehat{f}-\frac{1}{2\hat{%
\tau}}g_{ij}|^{2}+ |\widehat{R}_{ab}+\widehat{D}_{a}\widehat{D}_{b}\widehat{f}-\frac{1}{2\hat{%
\tau}}g_{ab}|^{2}](4\pi \hat{\tau})^{-n}e^{-\widehat{f}}dV.$
 We write  for $\mathbf{%
\tilde{D},}$ rescaling correspondingly the functions $\widehat{f}(\chi
)\rightarrow \tilde{f}(\chi ),\hat{\tau}(\chi )\rightarrow $ $\tilde{\tau}%
(\chi ),$
$$\frac{\partial g_{ij}}{\partial \chi } = -2\tilde{R}_{ij},\ \frac{\partial
g_{ab}}{\partial \chi }=-2\tilde{R}_{ab},\
\ \frac{\partial \tilde{f}}{\partial \chi } = -\tilde{\bigtriangleup}\tilde{%
f}+\left| \mathbf{\tilde{D}}\tilde{f}\right| ^{2}-\ _{s}\tilde{R}+\frac{2n}{%
\tilde{\tau}},\frac{\partial \tilde{\tau}}{\partial \chi }=-1$$
 and (for another functional, $\mathcal{\tilde{W})}$
$ \frac{\partial }{\partial \chi }\mathcal{\tilde{W}}(\mathbf{g}(\chi )\mathbf{%
,}\tilde{f}(\chi ),\tilde{\tau}(\chi ))=2\int_{\mathcal{V}}\tilde{\tau}[|%
\mathbf{\tilde{R}}_{\alpha \beta }+\mathbf{\tilde{D}}_{\alpha }\mathbf{%
\tilde{D}}_{\beta }\tilde{f}-\frac{1}{2\tilde{\tau}}\mathbf{\tilde{g}}%
_{\alpha \beta }|^{2}](4\pi \tilde{\tau})^{-n}e^{-\tilde{f}}dV.$

In the above formulas, we introduce the distorting relations
\begin{eqnarray}
\tilde{\nabla}&=&\ ^{F}\mathbf{D-}\ _{\nabla }^{F}\mathbf{Z,\ }\widehat{%
\mathbf{D}}=\ ^{F}\mathbf{D-}\ ^{F}\widehat{\mathbf{Z}},\ \mathbf{\tilde{D}}%
=\ ^{F}\mathbf{D-}\ ^{F}\mathbf{\tilde{Z}},   \label{dist3a}\\
\tilde{\Delta} &=&\mathbf{\tilde{D}}_{\alpha }\ \mathbf{\tilde{D}}^{\alpha
}=\ ^{F}\Delta +\ ^{Z}\tilde{\Delta},  \label{dist3b} \\
\ ^{F}\Delta  &=&\ ^{F}\mathbf{D}_{\alpha }\ ^{F}\mathbf{D}^{\alpha },\ \
^{Z}\tilde{\Delta}=\ ^{F}\mathbf{\tilde{Z}}_{\alpha }\ ^{F}\mathbf{\tilde{Z}}%
^{\alpha }-[\ ^{F}\mathbf{D}_{\alpha }(\ ^{F}\mathbf{\tilde{Z}}^{\alpha })+\
^{F}\mathbf{\tilde{Z}}_{\alpha }(\ ^{F}\mathbf{D}^{\alpha })];  \nonumber \\
\mathbf{\tilde{R}}_{\ \beta \gamma } &=&\ ^{F}\mathbf{R}_{\ \beta \gamma }-\
^{F}\mathbf{\tilde{Z}}ic_{\beta \gamma },\ \ _{s}\tilde{R}=\ _{s}^{F}R-%
\mathbf{g}^{\beta \gamma }\ ^{F}\mathbf{\tilde{Z}}ic_{\beta \gamma }=\
_{s}^{F}R-\ _{s}^{F}\mathbf{\tilde{Z}},  \nonumber \\
\ _{s}^{F}\mathbf{\tilde{Z}} &=&\mathbf{g}^{\beta \gamma }\ ^{F}\mathbf{%
\tilde{Z}}ic_{\beta \gamma }=\ _{h}^{F}\tilde{Z}+\ _{v}^{F}\tilde{Z},\
_{h}^{F}\tilde{Z}=g^{ij}\ ^{F}\mathbf{\tilde{Z}}ic_{ij},\ _{v}^{F}\tilde{Z}%
=h^{ab}\ ^{F}\mathbf{\tilde{Z}}ic_{ab};  \nonumber \\
\ _{s}^{F}R &=&\ _{h}^{F}R+\ _{v}^{F}R,\ \ _{h}^{F}R:=g^{ij}\ ^{F}R_{ij},\
_{v}^{F}R=h^{ab}\ ^{F}R_{ab},  \nonumber
\end{eqnarray}%
resulting in the equations from the conditions of theorem.
\end{proof}

Ricci flows with  $\nabla ,$ $\mathbf{\ }\widehat{\mathbf{D}}$ and $\mathbf{%
\tilde{D}}$ are characterized by respective thermodynamic values, see
section 5 in \cite{gper1} and, for metric compatible Finsler spaces, Refs. %
\cite{vric4,vric6}. Such constructions can be noholonomically deformed into metric noncompatible configurations.

In order to provide a statistical analogy, we consider a partition function $%
Z=\int \exp (-\beta E)d\omega (E)$ for the canonical ensemble at temperature
$\beta ^{-1}$ being defined by the measure taken to be the density of states $%
\omega (E).$ The thermodynamical values are computed in standard  form for
the  average energy, $\ \left\langle E\right\rangle :=-\partial \log
Z/\partial \beta ,$ the entropy $S:=\beta \left\langle E\right\rangle +\log Z
$ and the fluctuation $\sigma :=\left\langle \left( E-\left\langle
E\right\rangle \right) ^{2}\right\rangle =\partial ^{2}\log Z/\partial \beta
^{2}.$

\begin{theorem}
\label{theveq} Any family of Finsler geometries for which the conditions of
Theorem \ref{2theveq} are satisfied is characterized by thermodynamic values
{\small
\begin{eqnarray}
\left\langle\ ^{F}E\right\rangle &=&-\tilde{\tau}^{2}\int_{\mathcal{V}%
} (\ _{s}^{F}R+|\ ^{F}\mathbf{D}\tilde{f}|^{2}-\frac{n}{\widehat{\tau }}) \tilde{\mu}\ dV, \label{thermodv} \\
\ ^{F}S &=& -\int_{\mathcal{V}} [\tilde{\tau} (\ _{s}^{F}R+|\ ^{F}%
\mathbf{D}\tilde{f}|^{2}) +\tilde{f}-2n ] \tilde{\mu}\ dV, \nonumber \\
&&\ ^{F}\sigma = 2\ \tilde{\tau}^{4}\int_{\mathcal{V}}[\ ^{F}\mathbf{R}%
_{\alpha \beta }-\ ^{F}\mathbf{\tilde{Z}}ic_{\alpha \beta }+
 (\ ^{F}\mathbf{D}_{\alpha }-\ ^{F}\mathbf{\tilde{Z}}_{\alpha })(\
^{F}\mathbf{D}_{\beta }-\ ^{F}\mathbf{\tilde{Z}}_{\beta })\tilde{f}\\ 
&&-%
\frac{1}{2\tilde{\tau}}\mathbf{\tilde{g}}_{\alpha \beta }|^{2}]\tilde{\mu}\
dV \nonumber
\end{eqnarray}
}
\end{theorem}

\begin{proof}
There are two possibilities to prove this theorem. The first one is to use
the partition function $\breve{Z}=\exp \left\{ \int_{\mathcal{V}}[-\breve{f}+n]~\breve{\mu}%
dV\right\}$
and compute values (\ref{thermodv}) using methods from \cite{gper1,caozhu},
changing $\nabla \rightarrow \ ^{F}\mathbf{D}$  and rescaling $\breve{f}\rightarrow \tilde{f}$ and $%
\breve{\tau}\rightarrow \tilde{\tau}$ (such a rescaling is useful if we
wont to compare thermodynamical values for different Finsler connections). A
similar proof is possible if metric compatible Finsler connections are used.
For instance, considering  $\mathbf{\tilde{D}}\rightarrow \ ^{F}\mathbf{D}$
and $\tilde{Z}=\exp \left\{ \int_{\mathcal{V}}[-\tilde{f}+n]~\tilde{\mu}%
dV\right\} ,$ we compute \cite{vric4,vric6}
{\small
\begin{eqnarray*}
\left\langle \tilde{E}\right\rangle  &=&-\tilde{\tau}^{2}\int_{\mathcal{V}%
}\left( \ _{s}\tilde{R}+|\mathbf{\tilde{D}}\tilde{f}|^{2}-\frac{n}{\widehat{%
\tau }}\right) \tilde{\mu}\ dV,\\ \tilde{S} &=&-\int_{\mathcal{V}}\left[ \tilde{\tau}\left( \ _{s}\tilde{R}+|%
\mathbf{\tilde{D}}\tilde{f}|^{2}\right) +\tilde{f}-2n\right] \tilde{\mu}\ dV,
\\
\tilde{\sigma} &=&2\ \tilde{\tau}^{4}\int_{\mathcal{V}}[|\mathbf{\tilde{R}}%
_{\alpha \beta }+\mathbf{\tilde{D}}_{\alpha }\mathbf{\tilde{D}}_{\beta }%
\tilde{f}-\frac{1}{2\tilde{\tau}}\mathbf{\tilde{g}}_{\alpha \beta }|^{2}]%
\tilde{\mu}\ dV.
\end{eqnarray*}%
}
Introducing distortions (\ref{dist3a}) and (\ref{dist3b}) into the
thermodynamical values for $\mathbf{\tilde{D}}$, we generate analogous
thermodynamical values (\ref{thermodv}) for $\ ^{F}\mathbf{D.}$  \
 \end{proof}
The theorems and conclusions provided in this section can be formulated
and proved separately on $h$- and $v$--subspaces of a nonholonomic manifold $%
\mathbf{V}$ and/or a tangent bundle $TM.$ Some geometric and physical models with the Akbar-Zadeh curvature or other ''preferred'' Finsler connection (Berwald, Chern types etc) can be more/less/equivalent to alternative ones, but generated by the same $F.$ An exact answer is possible if a value $F$ is fixed following certain geometric/physical arguments.

\end{document}